%% file: ms.tex
\documentclass[runningheads]{llncs}
\usepackage[utf8]{inputenc}
\usepackage[T1]{fontenc}
\usepackage{amsmath,amssymb}
\usepackage{pgf}
\usepackage{tikz}
\usepackage{subfigure}

\usetikzlibrary{arrows,automata,shapes,petri,backgrounds}
\usetikzlibrary{positioning}
\input{bmacros}

\begin{document}
%\title{Kleene Theorems for Free Choice Automata over Distributed Alphabets}
\title{Kleene Theorems for Free Choice Nets Labelled with Distributed Alphabets}
%Alphabets\thanks{Supported by CPDA \& SGNF grants.}}
%Alphabets\thanks{Supported by IITDH/R\&D/SGNF/2018/008.}}
%\titlerunning{Kleene theorems for Free choice Automata over Distributed Alphabets}
%\titlerunning{Kleene theorems for Free choice nets Labelled with Distributed Alphabets}
\author{Ramchandra Phawade}
\authorrunning{Ramchandra Phawade}
 
\institute{ Indian Institute of Technology Dharwad,
%Dharwad \textit{580011},  
India\\
\email{prb@iitdh.ac.in}\\
}

\maketitle
%\linenumbers
\begin{abstract}
 
%We provided~\cite{PL14} expressions for free choice nets having  
We provided (PNSE'2014) expressions for free choice nets having  
\emph{distributed choice property} which makes the nets \emph{direct
product} representable. 
%In a recent work~\cite{P16}, we gave equivalent syntax for a larger class of  
In a recent work (PNSE'2016), we gave equivalent syntax for a larger class of  
free choice nets obtained by dropping distributed choice property.

In both these works, the classes of free choice nets were
restricted by a \emph{product condition} on the set of final markings.
In this paper we do away with this restriction and give expressions
for the resultant classes of nets which correspond to \emph{free choice
synchronous products and Zielonka automata}. 
For free choice nets with distributed choice property, 
we give an alternative characterization using properties  
checkable in polynomial time.   
 
Free choice nets we consider are $1$-bounded, S-coverable,
and are labelled with distributed alphabets, where S-components of
the associated S-cover respect the given alphabet distribution.  
%\end{abstract}
 
\end{abstract}
 
\section{Introduction}  
\label{lab:sec:intro}
\input{intro}
\section{Preliminaries} 
\label{lab:sec:prlm}
\input{prelim}

%\section{Nets} 
\label{lab:sec:nets}
\input{nets}

\section{Product systems}
\label{lab:sec:ps}
\input{ps}

\section{Nets and Product systems}
\label{sec-nets}
\input{n2p}

\section{Expressions} 
\label{lab:sec:expr} 
\input{expr}

\section{Connected Expressions and Product Systems}
\label{lab:section-exp-to-prod}
\input{e2p}

\section{Conclusion}
\label{lab:sec:concl}
\input{concl}

\bibliographystyle{splncs04}
\bibliography{ms}
 
%\newpage
%\appendix
%\section*{Appendix}
%\addcontentsline{toc}{section}{Appendices}
%\renewcommand{\thesubsection}{\Alph{subsection}}
%  
% \subsection{Proofs from Section~\ref{lab:sec:ps}}
%\label{lab:app:b1}
% \input{appendix-b1.tex}
%
%\subsection{Partitions of Derivatives of Regular Expressions}
%\label{lab:app:b}
%\input{appendix-b.tex}
% 
%\subsection{Proofs from Section~\ref{lab:sec:expr}}
%%\ref{lab:sec:expr}}
%\label{lab:app:c0}
%\input{appendix-c0.tex}
% 
%
%\subsection{Proofs from Section~\ref{lab:section-exp-to-prod}}
%\label{lab:app:c}
% \input{appendix-c.tex}
%
\end{document}

%% file: bmacros.tex
\newcommand{\Nat}{\mathbb{N}}
\renewcommand{\epsilon}{\varepsilon}

\newcommand{\calU}{\mathcal{U}}

\newcommand{\calA}{\mathcal{A}}
\newcommand{\calB}{\mathcal{B}}
\newcommand{\calC}{\mathcal{C}}
\newcommand{\calD}{\mathcal{D}}

\newcommand{\calG}{\mathcal{G}}

\newcommand{\name}{\textsf}

\newcommand{\concept}{\textbf}
\newcommand{\class}{\textsc}

\newcommand{\jour}{\textsl}

\newcommand{\cupover}{\displaystyle \bigcup}

\newcommand{\prodover}{\displaystyle \Pi}

\newcommand{\syncshuf}{\|}

\newcommand{\zero}{0}

\newcommand{\If}{\mbox{~if~}}
\newcommand{\Iff}{\mbox{~iff~}}

\newcommand{\Forall}{\mbox{~for~all~}}

\newcommand{\Otherwise}{\mbox{~otherwise}}
\newcommand{\Where}{\mbox{~where~}}
\newcommand{\rmand}{\mbox{~and~}}

\newcommand{\postdecomp}{\jour{postdecomp}}

\newcommand{\step}[1]{\xrightarrow{#1}}

\newcommand{\To}{\Rightarrow}
\newcommand{\Step}[1]{\mbox{$\stackrel{#1}{\To}$}}

\newcommand{\preset}[1]{\mbox{${ }^{\bullet}#1$}}
\newcommand{\postset}[1]{\mbox{$#1 \, { }^{\bullet}$}}%

\newcommand{\post}[1]{\mbox{$\jour{post}(#1)$}}
\newcommand{\pre}[1]{\mbox{$\jour{pre}(#1$)}}

\newcommand{\project}{\mbox{$\downarrow$}}

\newcommand{\ldot}{{\rm <}\kern-0.37em{\raisebox{.25ex}{\bf .}}\kern0.375em}

\newcommand{\fsync}{\jour{fsync}}
\newcommand{\Init}{\jour{Init}}
\newcommand{\Lang}{\jour{Lang}}
\newcommand{\Part}{\jour{Part}}
\newcommand{\Pref}{\mathit{Pref}}
\newcommand{\Suf}{\mathit{Suf}}
\newcommand{\Der}{\jour{Der}}
\newcommand{\pairing}{\jour{pairing}}

\newcommand{\cable}{\jour{cable~}}
\newcommand{\cables}{\jour{cables}}

\newcommand{\globals}{\jour{globals}}
\newcommand{\psglobals}{\jour{PS-globals}}

\newcommand{\matching}{\jour{matching}}
\newcommand{\psmatchings}{\jour{PS-matchings}}

\newcommand{\ce}{\jour{CE}}
\newcommand{\cepairings}{\jour{CE-pairings}}
\newcommand{\cecables}{\jour{CE-cables}}
 
\newcommand{\sce}{\jour{SCE}}
\newcommand{\scepairings}{\jour{SCE-pairings}}
\newcommand{\scecables}{\jour{SCE-cables}}

\newcommand{\ES}{\jour{ES}}

%% file: intro.tex
%\section{Introduction}  
%\label{lab:sec:intro}
%Free choice nets form an important subclass of Petri nets. %, having  pleasant theory~\cite{ThiVos84,DesEsp95}.  
%From verification point of view some of the advantages of $1$-bounded free choice  nets are: checking liveness is in \class{PTIME}~\cite{EspSil92,Des92}, checking deadlock is \class{NP}-complete~\cite{ChEsPa95} and reachability problem is \class{PSPACE}-complete~\cite{ChEsPa95}.  All these problems are \class{PSPACE}-complete for $1$-bounded  nets~\cite{ChEsPa95}. 
 
There are several different notions of acceptance to define languages
for labelled place transition Petri nets, depending on restrictions on  
labelling and \emph{final} markings~\cite{Pet76}.
The language of a place transition net with an initial marking and a finite set of final markings,  
is called $L$-type language~\cite{Jan87}. 
One goal of this work is to give syntax of expressions for $L$-type
languages for various subclasses of $1$-bounded, free choice nets labelled with distributed
alphabets.
%Zielonka's theorem~\cite{} shows that $1$-bounded nets can be seen as Zielonka automata.  
%For labelled nets, it is required that labelling should be such that
%concurrency in the behaviour is always between independent
%actions~\cite{}.   
One advantage of using distributed alphabet is that we can see free choice nets as products  
of automata~\cite{Muk11}, enabling us to write expressions for the nets using components.   
This also enables us to compare expressiveness of nets and products of automata.  
Three kinds of formulations of automata over distributed alphabets, in
the increasing order of expressiveness: direct products, synchronous products,  
and asynchronous products are described in \cite{Muk11}. In the present paper\footnote{A  
preliminary version of this paper appeared at $14$th PNSE workshop,
held at Bratislava~\cite{P18}.},  
we present a hirearchy of $1$-bounded free choice nets like automata
over distributed alphabets, and also introduce a fourth product automata in the  
current hierarchy which is utilized to get the syntax. 
In this hierarchy, there are four kinds of free choice nets labelled
over distributed alphabets. Two out of these four classes
were introduced earlier~\cite{PL14,PL15,P16}. Two new classes of systems
are given in this work. To understand the complete hierarchy and their
relations to other formalisms like expressions and automata over
distributed alphabets we invite the reader to read these earlier works
~\cite{PL14,PL15,P16}. 
 
We use product automata to get expressions for the
Free choice nets, and give correspondences for all these three formalisms  
for various classes. This kind of correspondence has been used
in concurrent code generation for discrete event systems~\cite{IorAnt12}.
 
We construct  
expressions for $L$-type languages of free choice nets via \emph{free choice Zielonka automata}.  
 
\begin{figure}[!htbp]
    \centering
    \begin{minipage}{.397\textwidth}
        \centering
        \input{./lsfc-nondp-fig1}
        \caption{Free Choice Net without distributed choice}
        \label{lsfc-nondp-fig1}
    \end{minipage}%
    \begin{minipage}{0.73\textwidth}
        \centering
	\fbox{\input{./lsfc-dp-fig3}}
	\caption{S-cover of the net in Fig.~\ref{lsfc-nondp-fig1}} 
	\label{lsfc-dp-fig3}
    \end{minipage}
\end{figure}

Consider the net $N$ of Figure~\ref{lsfc-nondp-fig1}  
%with initial marking $\{r_1,s_1\}$.  
with 
%Let its set of final markings be
$G=\{\{r_1,s_1\},\{r_2,s_2\}\}$ as its set of final markings, with 
its decomposition into %S-components (finite state machines) is also
finite state machines in Figure~\ref{lsfc-dp-fig3}.  
Because this net is decomposable into state machines~\cite{Hac72,DesEsp95}, its markings
%can be written in tuple form, where a place comes from one
%s-component: for example $G=\{(r_1,s_1),(r_2,s_2)\}$. 
can be written in tuple form, where each s-component has a place in
the tuple: for example $G=\{(r_1,s_1),(r_2,s_2)\}$. 
%can be written in tuple form, where each s-component has a place in
%the tuple: for example $G=\{(r_1,s_1),(r_2,s_2)\}$. 
For the final marking $(r_1,s_1)$, its language can be
expressed by $\fsync((ab+ac)^*, (ad+ae)^*)$~\cite{PL14,PL15}.
Similarly, for the final marking $(r_2,s_2)$ the language equivalent
expression can be given by $\fsync((ab+ac)^*a, (ad+ae)^*a)$.
In general, if the places involved in the final markings form a
\emph{product}~\cite{PL14}, then its language is specified by taking product of 
component expressions, %(called connected expressions developed in \cite{PL14}),  
using \emph{free choice Zielonka automata with product-acceptance}~\cite{P16}
as intermediary. %  with additional properties to capture free choice ness).  
Even though $r_1$ and $s_2$ participate in final markings,
marking $\{r_1,s_2\}$ does not belong to $G$, hence 
set $G$ do not form a product.
The language $L$ of net system $(N,G)$ can be described by,  
%For the net $N$ with $G$ as final markings, its language can be described by 
$\fsync((ab+ac)^*, (ad+ae)^*)+\fsync((ab+ac)^*a, (ad+ae)^*a)$.
The key idea is ability to express the language of a net as
the union of languages of nets complying with the product condition on final set of
markings. This closure under union may not be always possible for
restricted classes of languages defined over a distributed alphabet. For
example, the union of \emph{direct product} languages $L_1=\{ca,cb\}$ and
$L_2=\{caa,cbb\}$ defined over $\Sigma_1=\{c,a\}$ and $\Sigma_2=\{c,b\}$  
respectively, is not expressible as a \emph{direct product language}.   
But this language is accepted by \emph{synchronous products}: the direct
products extended with subset-acceptance~\cite{Muk11}.  
 
For the restricted class of direct product representable free choice nets,  
with its set of final markings having product condition-we  
gave expressions via product systems with matchings
(matched states of product system correspond to places of a cluster in net)
and product-acceptance~\cite{PL14,PL15}.  
As a second goal, we develop syntax for free choice nets with
distributed choice, now extended with subset-acceptance.
For a net in this class also, its language can be  
expressed as the union of languages accepted by product system with
matchings and product-acceptance.  
This union is accepted by product systems with matching and extended
with subset-acceptance (\emph{free choice synchronous products}).  
As a third contribution, we develop an alternate characterization
of this class of nets, 
%alternate characterization, using the restriction of
%product-acceptance on free choice Zielonka automata.  
via \emph{free choice Zielonka automata with
product-moves}. 
 
Language equivalent expressions for $1$-bounded nets have been  
given by Grabowski~\cite{Gra81}, Garg and Ragunath~\cite{GarRag92} 
and other authors~\cite{Lod06a}, where renaming operator has been used
in the syntax to disambiguate synchronizations. We have chosen to not
use this operator and to exploit the S-decompositions of nets instead. 
The syntax for smaller sublclasses of nets such \emph{marked graphs} and
\emph{free choice nets with initial markings as feedback vertex set} has
been given earlier ~\cite{LMP11,Pha-thesis}. 
 
\medskip 
\noindent{\emph{Organization of paper}}.
%Rest of the paper is organized as follows. 
In the next section, we begin with preliminaries on distributed alphabets  
and nets. % with their properties. 
In Section~\ref{lab:sec:ps} we define product systems with  
globals and subset-acceptance, and show that their  
languages can be expressed as the union of languages accepted by  
product systems with globals and product-acceptance.  
These product systems are used as intermediary to get expressions for nets and vice versa.
The following section relates these product systems to nets.  
In Section~\ref{lab:sec:expr} we develop syntax of expressions for product systems
with subset acceptance, and next section %Section~\ref{lab:section-exp-to-prod}
establishes the correspondence between various
classes of product systems and expressions.  
In the last section we conclude, with an overview of established correspondences
between all three formalisms.

%% file: lsfc-nondp-fig1.tex
%% lsfc nondp net --trace-labelled 
\begin{tikzpicture}[node distance=2cm,->,>=stealth',scale=0.8]

  %\tikzstyle{place}=[circle,thick,draw=blue!75,fill=blue!20,minimum size=6mm]
  \tikzstyle{place}=[circle,thick,draw=blue!75,fill=blue!10,minimum size=4mm]
    \tikzstyle{red place}=[place,draw=red!75,fill=red!20]
      %\tikzstyle{transition}=[rectangle,thick,draw=black!75, fill=black!20,minimum size=4mm]
      \tikzstyle{transition}=[rectangle,thin,draw=black!75, fill=black!10,minimum size=4mm]

    \tikzstyle{every label}=[red]

 \begin{scope}
      \node [place,tokens=1] (p1)    at(0.91,3)  [label=above:$r_1$]      {};
      \node [place,tokens=1] (p2)    at(4.1,3)  [label=above:$s_1$]      {};
      \node [place]          (p3)    at(0,1)  [label=below:$r_2$]      {};
      \node [place]          (p4)    at(1.5,1)  [label=below:$r_3$]      {};
      \node [place]          (p5)    at(4,1)  [label=below:$s_2$]      {};
      \node [place]          (p6)    at(5,1)  [label=below:$s_3$]      {};

      \node [transition] (t1) at(1,2)       {$a$}
       edge [pre]                             (p1)
       edge [pre]                             (p2)
       edge [post]                            (p3)
       edge [post]                            (p5);

      \node [transition] (t2) at(4,2)       {$a$}
       edge [pre]                             (p1)
       edge [pre]                             (p2)
       edge [post]                            (p4)
       edge [post]                            (p6);

      \node [transition] (t3) at(-0.5,0)       {$b$}
       edge [pre]                 (p3)
       edge [post,bend left]               (p1);

      \node [transition] (t4) at(2,0)       {$c$}
       edge [pre]                 (p4);
 
      \node [transition] (t5) at(3,0)       {$d$}
       edge [pre]                 (p5);
 
      \node [transition] (t6) at(5.5,0)       {$e$}
       edge [pre]                 (p6)
       edge [post,bend right]               (p2);
 
      \draw[->,black,rounded corners] (t4)--(2,-0.5)--(-1,-0.5)--(-1,3)--(p1);  
      \draw[->,black,rounded corners] (t5)--(3,-0.5)--(6,-0.5)--(6,3)--(p2);  
 
%\node[right=0cm,text width=6cm,font=\footnotesize] at (-1.5,-1.3){
%Figure($1$): T-system- live, 1-bdd, and S-component decomposable.
%Figure($2$): Trace-labelled FC-system having S-cover, not dp-representable
%};

%\node[right=0.5cm,text width=6cm,font=\footnotesize] at (3,2){
%Only S-component decomposition is:~
%$p_1ap_3 b p_1 \rmand p_1 a p_4 b p_1$.
%Language of Direct Product \emph(DP) formed by these
%components is Lang(DP,(1,2),(1,2))=\{ab\}, 
%while language of net is Lang(N,\{1,2,\},\{1,2\})= \{abb\}.
%};
%\begin{pgfonlayer}{background}
%       \filldraw [line width=4mm,join=round,black!7]
%       (-1,-0.7)  rectangle (6,3.5);
%\end{pgfonlayer}
\end{scope}
%\label{non-dp-net} 
\end{tikzpicture}
 

%% file: lsfc-dp-fig3.tex
%% lsfc nondp net --trace-labelled 
\begin{tikzpicture}[node distance=2cm,->,>=stealth',scale=0.80]

  %\tikzstyle{place}=[circle,thick,draw=blue!75,fill=blue!20,minimum size=4mm]
  \tikzstyle{place}=[circle,thick,draw=blue!75,minimum size=4mm]
    \tikzstyle{red place}=[place,draw=red!75,fill=red!20]
      %\tikzstyle{transition}=[rectangle,thick,draw=black!75, fill=black!20,minimum size=4mm]
      \tikzstyle{transition}=[rectangle,thin,draw=black!75,minimum size=1mm]

    \tikzstyle{every label}=[red]

 \begin{scope}
      \node [place,tokens=1] (p1)    at(1,3)  [label=above:$r_1$]      {};
      \node [place,tokens=1] (p2)    at(3.8,3)  [label=above:$s_1$]      {};
      \node [place]          (p3)    at(0,1)  [label=below:$r_2$]      {};
      \node [place]          (p4)    at(1.5,1)  [label=below:$r_3$]      {};
      \node [place]          (p5)    at(3.3,1)  [label=below:$s_2$]      {};
      \node [place]          (p6)    at(4.8,1)  [label=below:$s_3$]      {};

      \node [transition] (t1) at(0.5,2)       {$a$}
       edge [pre]                             (p1)
%       edge [pre]                             (p2)
       edge [post]                            (p3);
%       edge [post]                            (p5);

      \node [transition] (t2) at(1.5,2)       {$a$}
       edge [pre]                             (p1)
%       edge [pre]                             (p2)
       edge [post]                            (p4);
%       edge [post]                            (p6);
 
      \node [transition] (t7) at(3.3,2)       {$a$}
%       edge [pre]                             (p1)
       edge [pre]                             (p2)
%       edge [post]                            (p3)
       edge [post]                            (p5);

      \node [transition] (t8) at(4.3,2)       {$a$}
%       edge [pre]                             (p1)
       edge [pre]                             (p2)
%       edge [post]                            (p4)
       edge [post]                            (p6);

      \node [transition] (t3) at(-0.5,0)       {$b$}
       edge [pre]                 (p3)
       edge [post,bend left]               (p1);

      \node [transition] (t4) at(2,0)       {$c$}
       edge [pre]                 (p4);
 
      \node [transition] (t5) at(2.8,0)       {$d$}
       edge [pre]                 (p5);
 
      \node [transition] (t6) at(5.3,0)       {$e$}
       edge [pre]                 (p6)
       edge [post,bend right]               (p2);
 
      \draw[->,black,rounded corners] (t4)--(2,3)--(p1);  
      \draw[->,black,rounded corners] (t5)--(2.8,3)--(p2);  
 
      %\draw[->,black,rounded corners] (t4)--(2,-0.5)--(-1,-0.5)--(-1,3)--(p1);  
      %\draw[->,black,rounded corners] (t5)--(3,-0.5)--(6,-0.5)--(6,3)--(p2);  
%\node[right=0cm,text width=6cm,font=\footnotesize] at (-1.5,-1.3){
%Figure($1$): T-system- live, 1-bdd, and S-component decomposable.
%Figure($2$): Trace-labelled FC-system having S-cover, not dp-representable
%};

%\node[right=0.5cm,text width=6cm,font=\footnotesize] at (3,2){
%Only S-component decomposition is:~
%$p_1ap_3 b p_1 \rmand p_1 a p_4 b p_1$.
%Language of Direct Product \emph(DP) formed by these
%components is Lang(DP,(1,2),(1,2))=\{ab\}, 
%while language of net is Lang(N,\{1,2,\},\{1,2\})= \{abb\}.
%};
%\begin{pgfonlayer}{background}
%       \filldraw [line width=4mm,join=round,black!7]
%       (-1,-0.7)  rectangle (2.3,3.5);
%\end{pgfonlayer}
%\begin{pgfonlayer}{background}
%       \filldraw [line width=4mm,join=round,green!7]
%       (2.9,-0.7)  rectangle (6,3.5);
%\end{pgfonlayer}
\end{scope}
%\label{non-dp-net} 
\end{tikzpicture}
 

%% file: prelim.tex
%\section{Preliminaries} 
%\label{lab:sec:prlm}
$\Nat$ denotes the set of natural numbers including $0$.
Let $\Sigma$ be a finite alphabet  %such that $\{0,1\} \nsubseteq \Sigma$.
and $\Sigma^*$ be the set of all finite words over the alphabet $\Sigma$, including 
the empty word $\epsilon$. A \name{language} over an alphabet $\Sigma$ is 
a subset $L \subseteq \Sigma^*$. 
The \name{projection} of a word $w \in \Sigma^*$ to a  
set $\Delta \subseteq \Sigma$, denoted as $w\project_{\Delta}$,  
is defined by:
$\epsilon\project_{\Delta}=\epsilon$ and
\begin{math}
\begin{array}{ll}
(a\sigma)\project_{\Delta}= 
\begin{cases}
  a(\sigma\project_{\Delta}) & \text{if} ~a \in \Delta,\\
  \sigma\project_{\Delta}    & \text{if} ~a \notin \Delta.
\end{cases}
\end{array}
\end{math}

\noindent
Given languages $L_1,L_2,\ldots,L_m$, their 
\name{synchronized shuffle} $L = L_1 \syncshuf \dots \syncshuf L_m$ 
is defined as:  
$w \in L \Iff \Forall i \in \{1,\ldots,m\}, w\project_{\Sigma_i} \in L_i$.

\begin{definition}[Distributed Alphabet]
Let $Loc$ denote the set $\{1,2,\ldots,k\}$.
A \name{distribution} of $\Sigma$ over $Loc$ is 
a tuple of nonempty sets $(\Sigma_1,\Sigma_2,\ldots,\Sigma_k)$
with $\Sigma = \bigcup_{1\leq i\leq k}\Sigma _i$.
For each action $a \in \Sigma$,  its \name{locations}
are the set  $loc(a)=\{i \mid a \in \Sigma_i \}$.  
Actions $a \in \Sigma$ such that $|loc(a)|=1$ are called \name{local},
otherwise they are called \name{global}. 
\end{definition}  
A global action is global in the locations in which it occurs. 
For a set $S$ let $\wp(S)$ denote the set of all its susbets. 
For singleton sets like $\{p\}$, sometimes we may write it as $p$.
 
We will sometimes write $p$ instead of the singleton $\{p\}$. 
%For singleton sets like $\{p\}$, sometimes we may write it as $p$.
% depending on the context we may just write it as $p$.  
 
%We use the following characterization of direct product languages, 
%which appears in \cite{MR02,Muk11}.

Let $I=\{i_1,\ldots,i_m\} \subseteq \{1,\ldots,k\}$ be a set of indices
with $1 \leq i_1 < i_2 < \ldots < i_m \leq k$, 
and let $Z_1,\ldots,Z_k$ be finite sets. 
Then $\prodover_{i \in I} Z_i= \{(z_{i_1},\ldots,z_{i_m}) \mid z_{i_j}
\in Z_{i_j}, \forall j \in \{1,\ldots,m\}\}$. 

Let $Z=\prodover_{i \in Loc} Z_i$ and $z=(z_{1},\ldots,z_{k}) \in Z$. 
Then restriction of $z$ to $I$ is the subset of its components taken in  
the order given by $I$ i.e., $z\project I = (z_{i_1},\ldots,z_{i_m})$.
And its generalization 
$Z \project I = \{ (z_{i_1},\ldots,z_{i_m}) \mid~ \exists z \in Z ~\mbox{with}~
z\project I = (z_{i_1},\ldots,z_{i_m})\}$.

%% file: nets.tex
%\section{Nets} 
%\label{lab:sec:nets}
 
\subsection{Nets} 
 
%\subsection{Labelled nets over a distribution}
\begin{definition}      
A \name{labelled net} $N$ is a tuple $(S,T,F,\lambda)$, where $S$ is a
finite set of places, $T$ is a finite set (disjoint from $S$) of transitions  
labelled by the function $\lambda: T \to \Sigma$ and  
$F \subseteq (T \times S) \cup (S \times T)$ is the flow relation. 
\end{definition}

Elements of $S \cup T$ are called \name{nodes} of $N$.
Given a node $z$ of net $N$, set $\preset z = \{ x \mid (x,z) \in F \}$
is called \name{pre-set} of $z$   and
 $\postset z = \{ x \mid (z,x) \in F \}$ is called \name{post-set} of $z$.
 Given a set $Z$ of nodes of $N$,
let $\preset Z = \bigcup_{z \in Z} \preset z
\rmand 
 \postset Z = \bigcup_{z \in Z} \postset z$.
 
We only consider nets in which every transition 
has nonempty pre- and post-set.
For each action $a$ in $\Sigma$ let  
$T_a=\{ t \mid t \in T \rmand \lambda(t)=a\}$.
 
A path of net $N$ is a nonempty sequence $x_1 \ldots x_n$ of nodes of $N$  
where $(x_i,x_{i+1}) \in F$ for all $i$ in $\{1,\ldots,n-1\}$. We say
that this path leads from node $x_1$ to $x_n$. 
Net $N$ is said to be \name{connected} if for any two nodes $x$ and
$y$ there exists a path leading $x$ to $y$ or from $y$ to $x$. 
The net is \name{strongly connected} if for any two nodes $x$ and
$y$ there exists a path leading from $x$ to $y$ and a path from $y$ to $x$. 
 
A net is called an \name{S-net} \cite{DesEsp95} if for any transition $t$ we have 
$|\preset{t}|=1=|\postset{t}|$. 
 
%\subsubsection{Net Systems and their Languages}
A \name{marking} of a net $N$ is mapping $M:S \to \Nat$. 
At marking $M$, a place $p$ is said to be marked if $M(p) \geq 1$,   
and is said to be unmarked if $M(p)=0$. 
 
%We are only interested in \name{1-bounded} nets, where  
%a place is either marked or not marked. Hence we define a
%\name{marking} of a net as a subset of its places. 
 
\begin{definition}%[\nssac]
%A labelled net system is a tuple $(N,M_0,\calG)$ where 
%$N=(S,T,F,\lambda)$ is a labelled net with $M_0 \subseteq S$ as its initial marking
%and a set of markings $\calG \subseteq \wp(S)$.  
A labelled net system is a tuple $(N,M_0,\calG)$ where 
$N=(S,T,F,\lambda)$ is a labelled net; $M_0$ an initial marking;  
and a finite set of final markings $\calG$.  
 
\end{definition} 

%A transition $t$ is \name{enabled} in a marking $M$ if all places in its  
%pre-set are marked by $M$. In such a case, $t$ can be \name{fired} to  
%produce the new marking $M' = (M \setminus \preset{t}) \cup \postset{t}$ and,  
%we write this as $M \step{t} M'$ or $M \step{\lambda(t)} M'$.
%
%A \name{firing sequence} %(finite or infinite)
%$\lambda(t_1) \lambda(t_2) \dots$ is defined by composition of labels
%of transition sequence $t_1 t_2\ldots$ fired from $M_0$ i.e.,
%$M_0 \step{t_1} M_1 \step{t_2} \dots$. 
%In such a firing sequence, we say that $M_j$ is \name{reachable} from
%$M_i$, for every $i \leq j$.
%We say a net system $(N,M_0)$ is \name{live} if, for every reachable marking $M$ 
%and every transition $t$, there exists a marking $M'$ reachable from $M$
%which enables $t$. 
% 
 
%\noindent 
A transition $t$ is \name{enabled} at a marking $M$ if all places in its  
pre-set are marked by $M$. In such a case, $t$ can be \name{fired or
occurs} at $M$, to  produce the new marking $M'$ which is defined as : 
for each place $p$ in $S$, $M'(p)=M(p) + F(t,p) - F(p,t)$, 
where $F(x,y)=1$ if $(x,y) \in F$ and $0$ otherwise. 
We write this as $M \step{t} M'$ or $M \step{\lambda(t)} M'$.

%A \name{firing or occurrence sequence} %(finite or infinite)
%$\lambda(t_1) \lambda(t_2) \dots$ is defined by composition of labels
%of transition sequence $t_1 t_2\ldots$ fired from $M_0$ i.e.,
%$M_0 \step{t_1} M_1 \step{t_2} \dots$. 
%In such a firing sequence, we say that $M_j$ is \name{reachable} from
%$M_i$, for every $i \leq j$.
 
For some markings $M_0,M_1,\ldots,M_n$  if we have 
$M_0 \step{t_1} M_1 \step{t_2} \ldots \step{t_n} M_n$, then
the sequence $\sigma=t_1 t_2 \ldots t_n$ is called \name{occurrence or firing
sequence}. We write $M_0 \step{\sigma} M_n$ and call $M_n$ the marking
\name{reached} by $\sigma$. This includes an empty transition sequence
$\epsilon$. For each marking $M$ we have $M \step{\epsilon} M$.  
We write $M \step{*} M'$ and call $M'$ \name{reachable} from $M$ if it is  
reached by some occurrence sequence $\sigma$ from $M$. 

A net system $(N,M_0,\calG)$ is called \name{$1$-bounded} if for every place
$p$ of the net and every reachable marking $M$, we have $M(p) \leq 1$. 
Any marking $M$ of a $1$-bounded net can be alternately represented by the subset of places which are marked at $M$. 
In this paper, we consider only $1$-bounded nets.  

We say a net system $(N,M_0)$ is \name{live} if, for every reachable marking $M$ 
and every transition $t$, there exists a marking $M'$ reachable from $M$
which enables $t$.

\begin{definition}
For a labelled net system $(N,M_0,\calG)$, its
\name{language} is defined as
$Lang(N,M_0,\calG)=\{ \lambda(\sigma) \in \Sigma^{*} \mid 
	\sigma \in T^{*}~\mbox{and}~ M_0 \step{\sigma}M, 
	~\mbox{for some}~M \in \calG\}$.
\end{definition} 
 
\subsubsection{Net Systems and its components}
First we define subnet of a net. 
 
%Our notion of component does not require strong connectedness and so  
%it is different from notion of $S$-component in \cite{DesEsp95}, and therefore
%our notion of $S$-cover also differs from theirs.
 
%\begin{definition}      
%Let $X$ be a set of nodes of net $N=(S,T,F)$.
%Then the triple $N'=(S \cap X, T \cap X, F \cap (X \times X))$ is
%a \name{subnet} of net $N$. It is said to be a subnet of $N$
%\name{generated} by nodes $X$ of $N$.  
% 
%Let $N'=(S \cap X, T \cap X, F \cap (X \times X))$ be a \name{subnet} of net 
%$N=(S,T,F)$, generated by a nonempty set $X$ of nodes of $N$. 
%Subnet $N'$ is called a \concept{component} of $N$ if, 
%\begin{itemize}  
%\item For each place $s$ of $X$, $\preset{s},\postset{s} \subseteq X$
%(the pre- and post-sets are taken in $N$),
%\item For all transitions $t \in T$, we have $|\preset t|=1=|\postset t|$
%($N'$ is an $S$-net \cite{DesEsp95}), 
%\item Under the flow relation, $N'$ is connected.
%%made up of strongly connected components. 
%\end{itemize} 
%A set $\mathcal{C}$ of components of net $N$ is called 
%\concept{S-cover} for $N$, if every place of the net belongs to some component of $\mathcal{C}$.  
%\end{definition}  
% 
 
Let $X$ be a set of nodes of net $N=(S,T,F)$.
Then the triple $N'=(S \cap X, T \cap X, F \cap (X \times X))$ is
a \name{subnet} of net $N$.  
Flow relation $F \cap (X \times X)$  is said to be \name{induced}  
by nodes $X$; and $N'$ is said to be a subnet of $N$ \name{generated}  
by nodes $X$ of $N$.  
 
We follow the convention that if $N'$ is a subnet of $N$ and $z$ is a node of $N'$ then 
$\preset{z}$ and $\postset{z}$ denote the pre-set and post-set taken
in $N$, i.e., $\preset z = \{ x \mid (x,z) \in F \}$ 
and $\postset z = \{ x \mid (z,x) \in F \}$.  
 
\begin{definition}      
 
%Let $N'=(S \cap X, T \cap X, F \cap (X \times X))$ be a \name{subnet} of net 
%$N=(S,T,F)$, generated by a nonempty set $X$ of nodes of $N$. 
Subnet $N'$ is called a \concept{component} of $N$ if, 
\begin{itemize}  
\item For each place $s$ of $X$,  
$\preset{s},\postset{s} \subseteq X$,
\item $N'$ is an $S$-net, %\cite{DesEsp95}, 
\item $N'$ is connected.
%made up of strongly connected components. 
\end{itemize} 
A set $\mathcal{C}$ of components of net $N$ is called 
\concept{S-cover} for $N$, if every place of the net belongs to some component of $\mathcal{C}$.  
\end{definition}  
 
Our notion of component does not require strong connectedness and so  
it is different from notion of $S$-component in \cite{DesEsp95},  
and therefore our notion of $S$-cover also differs from theirs.
 
A net is \name{covered by components or S-coverable} if it has an $S$-cover. 

Fix a distribution $(\Sigma_1,\Sigma_2,\ldots,\Sigma_k)$ of $\Sigma$.  
We define s-decomposition \cite{Hac72} of a net into sequential components.
% appears in several places for unlabelled nets. 
%starting with \cite{Hac72}.
Note that S-decomposition given here is for labelled nets unlike~\cite{Hac72,DesEsp95}  
and is different from \cite{PL14,P16,P18} also,
as it takes into account the initial marking of the net. 
%\begin{definition}      
%A labelled net $N=(S,T,F,\lambda)$ is called \concept{S-decomposable} if,
%there exists an S-cover $\mathcal{C}$ for $N$, such that
%for each $T_i=\{ \lambda^{-1}(a) \mid a \in \Sigma_i \}$, there exists 
%$S_i$ such that the subnet generated by $S_i \cup T_i$ is a component
%in $\mathcal{C}$.
%%$S_i$ such that the induced component $(S_i,T_i,F_i)$ is in $\mathcal{C}$.
%\label{s-decomposable} 
%\end{definition}  
%
 
\begin{definition}      
A labelled net system $(N,M_0,\calG)$ is called \concept{S-decomposable}  
if, there exists an S-cover $\mathcal{C}$ for net $N=(S,T,F,\lambda)$, such that for each  
$T_i= \bigcup_{a \in \Sigma_i} \lambda^{-1}(a)$, there exists $S_i \subseteq S$  
and the subnet generated by $S_i \cup T_i$ is a component in
$\mathcal{C}$,  
and the initial marking $M_0$ marks only one place of the component. 
\label{s-decomposable} 
\end{definition}  

%Now from S-decomposability we get an $S$-cover for net $N$, since
Now each S-decomposable net $N$ admits an $S$-cover, since  
there exist subsets $S_1,S_2,\ldots,S_k$ of places $S$, such that
 $S= S_1 \cup S_2 \cup \ldots S_k$ and 
 $~^{\bullet}S_i \cup S_i^{\bullet}=T_i$, such that
the subnet $(S_i,T_i,F_i)$  generated by $S_i$ and $T_i$ is an S-net,
where $F_i$ is the induced flow relation from $S_i$ and $T_i$.

%If a net $(S,T,F,\lambda)$ is $1$-bounded and S-decomposable
%then a \name{marking} can be written as a $k$-tuple from 
%its component places $S_1 \times S_2 \times \ldots \times S_k$.   

Note that, the initial marking, of a $1$-bounded and S-decomposable net system, 
marks exactly one place in each S-component of the given S-cover $S_1,S_2,\ldots,S_k$. 
At any reachable markings of such a net, the total number of tokens in an S-component 
remains contant~\cite{Hac72,DesEsp95}. Therefore, at any reachable marking $M$, 
each S-component has only on token, so at that marking only one place of that component is marked. 
Also, if we collect each place from an S-component we get back the marking of net. 
Hence, marking $M$ can be written as a $k$-tuple from its component places $S_1 \times S_2 \times \ldots \times S_k$.   
%If we do not enforce the ``direct product'' condition on the set of  
%final markings, given below, it is known \cite{Zie87,Muk11} that we  
%get a larger subclass of languages.  
 
We use a \emph{product} condition \cite{PL14} on the set of final markings of a net system  
which is known \cite{Zie87,Muk11} to restrict classes of languages. 

%\begin{definition}
%An \name{S-decomposable labelled net system with product-acceptance}
%%(\nspac)}  
%$(N,M_0,\calG)$ is an S-decomposable labelled net $N=(S,T,F,\lambda)$ 
%with initial marking  $M_0$ 
%and a set of markings $\calG \subseteq \wp(S)$, 
%which is a \name{product}:
%if $\langle q_1,q_2,\ldots q_k \rangle \! \! \in \calG$ and
%   $\langle q'_1,q'_2,\ldots q'_k \rangle \in \calG$ 
%then 
%   $\{q_1,q'_1\} \times \{q_2,q'_2\} \times \ldots \times \{q_k,q'_k\} 
%\subseteq \calG$.
%\end{definition} 
%
%

\begin{definition}
An \name{S-decomposable labelled net system}
$(N,M_0,\calG)$ is said to have \name{product-acceptance}
%(\nspac)}  
if its set of final markings $\calG$ satisfies \name{product condition}:
if $\langle q_1,q_2,\ldots q_k \rangle \in \calG$ and
   $\langle q'_1,q'_2,\ldots q'_k \rangle \in \calG$ 
then $\{q_1,q'_1\} \times \{q_2,q'_2\} \times \ldots \times \{q_k,q'_k\} \subseteq \calG$.
\end{definition} 

Let $t$ be a transition in $T_a$. Then by $S$-decomposability 
a pre-place and a post-place of $t$ belongs to each $S_i$ for all $i$ in $loc(a)$. 
Let $t[i]$ denote the tuple $\langle p,a,p' \rangle$ such that  
$(p,t), (t,p') \in F_i, ~\mbox{and}~p,p' \in P_i$
for all $i$ in $loc(a)$.

\subsection{Free choice nets and their properties}
 
Let $x$ be a node of a net $N$. The \name{cluster} of $x$, denoted
by $[x]$, is the minimal set of nodes containing $x$ such that
\begin{itemize}
\item if a place $s \in [x]$ then $s^{\bullet}$ is included in
$[x]$, and
\item if a transition $t \in [x]$ then $~^{\bullet}t~$ is
included in $[x]$.
\end{itemize}
 
For a cluster $C$, we denote its set of places by $S_{C}$,
and its set of transitions by $T_{C}$. 
 
The set of all $a$-labelled transitions along with places $r_1$ and
$s_1$ form a cluster of the net shown in Figure~\ref{lsfc-dcp-fig1}.

\begin{definition}[Free choice nets ~\cite{DesEsp95}]
A cluster $C$ is called
\name{free choice} (FC) if all transitions in $C$ have the same pre-set.
A net is called \name{free choice} if all its clusters
are free choice.
\end{definition}

%In a labelled net $N$, for a free choice cluster $C=(S_C,T_C)$
In a labelled net $N$, for a free choice cluster $C$
define the $a$-labelled transitions $C_a = \{ t \in T_C \mid \lambda(t)=a\}$.
If the net has an S-decomposition then %generated by $S_i$,
we associate a \name{post-product}
$\pi(t) = \prodover_{i \in loc(a)} (\postset{t} \cap S_i)$ 
with every such transition $t$.
This is well defined since in S-nets, 
every transition will have at most one post-place in $S_i$.
Let $\post{C_a} = \cupover_{t \in C_a} \pi(t)$.
%We also define the post-projection of the cluster
%$C_a[i] = \postset{C_a} \cap S_i$
%and the \name{post-decomposition}
%$\postdecomp(C_a) = \prodover_{i \in loc(a)} C_a[i]$.
%Clearly $\post{C_a} \subseteq \postdecomp(C_a)$.
Let %We also define the post-projection of the cluster
$C_a[i] = \postset{C_a} \cap S_i$
and %the \name{post-decomposition}
$\postdecomp(C_a) = \prodover_{i \in loc(a)} C_a[i]$.
Clearly $\post{C_a} \subseteq \postdecomp(C_a)$.
Sometimes, we may call $C_a[i]$ as \name{post-projection} of the
cluster $C$ with respect to label $a$ and location $i$.
Also, $\postdecomp(C_a)$ is called \name{post-decomposition} of cluster
$C$ with respect to label $a$.  
 
The following definition from \cite{PL15,PL14} 
is used to get direct product representability.  

\begin{definition}[distributed choice property]
An S-decomposable free choice net $N=(S,T,F,\lambda)$ is said to  
    have \name{distributed choice} property (DCP)  if,
for all $a$ in $\Sigma$ and for all clusters $C$ of $N$,
$\postdecomp(C_a) \subseteq \post{C_a}$.
\label{lab-dc} 
\end{definition}

%%%% free choice nets with dcp
 
\begin{figure}[!htbp]
    \centering
    \begin{minipage}{.399\textwidth}
        \centering
	\input{./lsfc-dcp-fig1}
        %\caption{Labelled Free Choice Net with distributed choice}
        \caption{Labelled Free Choice Net system with distributed choice}
        \label{lsfc-dcp-fig1}
    \end{minipage}%
    \begin{minipage}{0.73\textwidth}
        \centering
	%\fbox{\input{./ps-dcp-fig1}}
	\input{./ps-dcp-fig1}
	\caption{S-cover of the net in Fig.~\ref{lsfc-dcp-fig1}} 
	\label{ps-dcp-fig1}
    \end{minipage}
\end{figure}
 
%Example~\ref{lab:ex:fcns-pac} given in Appendix~\ref{lab:app:exfcnondcpnetpac}  
%shows a free choice net system without distributed choice and 
%having product-acceptance.
 
\begin{example}[Free choice net system without distributed choice and
with product-acceptance]
\label{lab:ex:fcns-pac}
Consider a distributed alphabet $\Sigma=(\Sigma_1=\{a,b,c\},\Sigma_2=\{a,d,e\})$  
%and the labelled net system $N$ %=(S,T,F,\lambda)$  
and the net system $N$ %=(S,T,F,\lambda)$  
shown in Figure~\ref{lsfc-nondp-fig1}, labelled over $\Sigma$. 
Its (only possible) S-cover having two S-components %$N_1$ and $N_2$
with sets of places $S_1=\{r_1,r_2,r_3\}$ and $S_2=\{s_1,s_2,s_3\}$ respectively, 
is given in Figure~\ref{lsfc-dp-fig3}.
For the cluster $C$ of $r_1$, we have the set of $a$-labelled transitions 
%$C_a=\{t_1,t_2\}$ with its post-projections $C_a[1]=\{r_2,r_3\}$ and 
$C_a=\{t_1,t_2\}$ with $C_a[1]=\{r_2,r_3\}$ and $C_a[2]=\{s_2,s_3\}$.  
So we get 
$\postdecomp(C_a)=\{(r_2,s_2),(r_2,s_3),(r_3,s_2),(r_3,s_3)\}$.  
 
%The post-products are $\pi(t_1)=\{(r_2,s_2)\}$ and 
As $\pi(t_1)=\{(r_2,s_2)\}$ and 
$\pi(t_2)=\{(r_3,s_3)\}$ so $\post{C_a}=\{(r_2,s_2),(r_3,s_3)\}$. 
Since $\postdecomp(C_a)\nsubseteq \post{C_a}$, 
this cluster does not have distributed choice, so
the net system does not have it. 

With the set of final markings
$\{(r_1,s_1),(r_1,s_2),(r_2,s_1),(r_2,s_2)\}$ satisfying product
condition, the language $L_p$ accepted by this net system is
$r^*[\epsilon+a+ab+ad]$ where
$r= (a(bd+db)+a(ce+ec))$. 
\end{example}  
 
%

%Example~\ref{exfcnondcpnetpac} given in Appendix~\ref{lab:app:exfcnondcpnetpac}  
%shows a free choice net system without distributed choice and 
%not having product-acceptance.
 
\begin{example}[Free choice net system without distributed choice and 
not satisfying product condition of the set of final markings]
\label{exfcnondcpnetpac}
%Final marking without product condition (ZA). 
Consider the net system of Example~\ref{lab:ex:fcns-pac} whose
underlying net is shown in Figure~\ref{lsfc-nondp-fig1}.
With set of final markings $\{(r_1,s_1),(r_2,s_2)\}$, which do not
satisfy product condition, the language $L_s$ accepted by this net system is
$r^*[\epsilon+a]$ where $r= (a(bd+db)+a(ce+ec))$. 
\end{example}

%Example~\ref{lab:ex:fcns-dcp-pac} given in Appendix~\ref{lab:app:exfcnondcpnetpac}  
%shows a free choice net system with distributed choice and product-acceptance.

\begin{example}[A net with distributed choice property and product
acceptance condition]  
\label{lab:ex:fcns-dcp-pac}
Consider the labelled net system $(N,(r_1,s_1),\calG))$ of
Figure~\ref{lsfc-dcp-fig1}, defined over distributed alphabet  
$\Sigma=(\Sigma_1=\{a,b,c\}, \Sigma_2=\{a,d,e\})$, 
and where $\calG=\{(r_1,s_1),(r_1,s_2),(r_2,s_1),(r_2,s_2)\}$ is 
the set of final markings satisfying product condition. 
Its two S-components with sets of places
$S_1=\{r_1,r_2,r_3\}$ and $S_2=\{s_1,s_2,s_3\}$, 
are shown in Figure~\ref{ps-dcp-fig1}.
For cluster $C$ of $r_1$, we have $C_a=\{t_1,t_2,t_3,t_4\}$, 
$C_a[1]=\{r_2,r_3\}$ and $C_a[2]=\{s_2,s_3\}$, hence
$postdecomp(C_a)=\{(r_2,s_2),(r_2,s_3),(r_3,s_2),(r_3,s_3)\}$.  
%The post-products for transitions are  
We have
$\pi(t_1)=\{(r_1,s_2)\}$,
$\pi(t_2)=\{(r_2,s_3)\}$,
$\pi(t_3)=\{(r_3,s_2)\}$ and 
$\pi(t_4)=\{(r_3,s_3)\}$.  \\
So $post(C_a)=\{(r_2,s_2),(r_2,s_3),(r_3,s_2),(r_3,s_3)\}$.  
Therefore, $postdecomp(C_a)= post(C_a)$. 
For all other clusters this holds trivially, because each of them
have only one transition and only one post-place, hence the net 
has distributed choice. 
Language $L_3$ accepted by the net system is 
$r^*[\epsilon+a+a(b+c)+a(d+e)]$ where
$r= (a(bd+db)+a(be+eb)+a(cd+dc)+a(ce+ec))$. 
\end{example}
 
\begin{example}[A net with distributed choice and subset-acceptance]  
%acceptance condition (\fcnsdcpsac)]  
\label{lab:ex:fcns-dcp-sac}
    Consider the net system of Example~\ref{lab:ex:fcns-dcp-pac},
with the underlying net shown in 
    Figure~\ref{lsfc-dcp-fig1} with set of final markings
$\{(r_1,s_1),(r_2,s_2)\}$.
The language $L_4$ accepted by this net system is
$r^*[\epsilon+a]$ where
$r= (a(bd+db)+a(be+eb)+a(cd+dc)+a(ce+ec))$. 
\end{example}

%% file: lsfc-dcp-fig1.tex
%% lsfc dcp net --trace-labelled 
\begin{tikzpicture}[node distance=2cm,->,>=stealth',scale=.919]

  \tikzstyle{place}=[circle,thick,draw=blue!75,fill=blue!10,minimum size=4mm]
    \tikzstyle{red place}=[place,draw=red!75,fill=red!20]
      \tikzstyle{transition}=[rectangle,thin,draw=black!75,
        			  fill=black!10,minimum size=4mm]

    \tikzstyle{every label}=[red]

 \begin{scope}
      \node [place,tokens=1] (r1)    at(1,3)  [label=above:$r_1$]      {};
      \node [place,tokens=1] (s1)    at(4,3)  [label=above:$s_1$]      {};
      \node [place]          (r2)    at(0,1)  [label=below:$r_2$]      {};
      \node [place]          (r3)    at(1.5,1)  [label=below:$r_3$]      {};
      \node [place]          (s2)    at(3.5,1)  [label=below:$s_2$]      {};
      \node [place]          (s3)    at(5,1)  [label=below:$s_3$]      {};

      \node [transition] (t1) at(1,2)        {$a$}
       edge [pre]                             (r1)
       edge [pre]                             (s1)
       edge [post]                            (r2)
       edge [post]                            (s2);
 
      \node [transition] (t11) at(2,2)       {$a$}
       edge [pre]                             (r1)
       edge [pre]                             (s1)
       edge [post]                            (r2)
       edge [post]                            (s3);

      \node [transition] (t21) at(3,2)       {$a$}
       edge [pre]                             (r1)
       edge [pre]                             (s1)
       edge [post]                            (r3)
       edge [post]                            (s2);
 
      \node [transition] (t2) at(4,2)       {$a$}
       edge [pre]                             (r1)
       edge [pre]                             (s1)
       edge [post]                            (r3)
       edge [post]                            (s3);

      \node [transition] (t3) at(-0.5,0)       {$b$}
       edge [pre]                 	   (r2)
       edge [post,bend left]               (r1);

      \node [transition] (t4) at(2,0)       {$c$}
       edge [pre]                 (r3);
 
      \node [transition] (t5) at(3,0)       {$d$}
       edge [pre]                 (s2);
 
      \node [transition] (t6) at(5.5,0)       {$e$}
       edge [pre]                 (s3)
       edge [post,bend right]               (s1);
 
      \draw[->,black,rounded corners] (t4)--(2,-0.5)--(-1,-0.5)--(-1,3)--(r1);  
      \draw[->,black,rounded corners] (t5)--(3,-0.5)--(6,-0.5)--(6,3)--(s1);  
 
%\node[right=0cm,text width=6cm,font=\footnotesize] at (-1.5,-1.3){
%Figure($1$): T-system- live, 1-bdd, and S-component decomposable.
%Figure($2$): Trace-labelled FC-system having S-cover, not dp-representable
%};

%\node[right=0.5cm,text width=6cm,font=\footnotesize] at (3,2){
%Only S-component decomposition is:~
%$p_1ap_3 b p_1 \rmand p_1 a p_4 b p_1$.
%Language of Direct Product \emph(DP) formed by these
%components is Lang(DP,(1,2),(1,2))=\{ab\}, 
%while language of net is Lang(N,\{1,2,\},\{1,2\})= \{abb\}.
%};
%\begin{pgfonlayer}{background}
%       \filldraw [line width=4mm,join=round,black!7]
%       (-1,-0.7)  rectangle (6,3.5);
%\end{pgfonlayer}
\end{scope}
%\label{non-dp-net} 
\end{tikzpicture}
 

%% file: ps-dcp-fig1.tex
%% s-cover of dcp fig1, with four transitions in each component.
 
\begin{tikzpicture}[node distance=2cm,->,>=stealth',scale=0.82]

  %\tikzstyle{place}=[circle,thick,draw=blue!75,fill=blue!20,minimum size=4mm]
  \tikzstyle{place}=[circle,thick,draw=blue!75,minimum size=4mm]
    \tikzstyle{red place}=[place,draw=red!75,fill=red!20]
      %\tikzstyle{transition}=[rectangle,thick,draw=black!75, fill=black!20,minimum size=1mm]
      \tikzstyle{transition}=[rectangle,thin,draw=black!75,minimum size=1mm]

    \tikzstyle{every label}=[red]

 \begin{scope}
      \node [place,tokens=1] (p1)    at(0.81,3)  [label=above:$r_1$]      {};
      \node [place,tokens=1] (p2)    at(4,3)  [label=above:$s_1$]      {};
      \node [place]          (p3)    at(0,1)  [label=right:$r_2$]      {};
      \node [place]          (p4)    at(1.5,1)  [label=below:$r_3$]      {};
      \node [place]          (p5)    at(3.5,1)  [label=right:$s_2$]      {};
      \node [place]          (p6)    at(5,1)  [label=below:$s_3$]      {};

      \node [transition] (t1) at(-0.3,2)       {$a$}
       edge [pre]                             (p1)
       edge [post]                            (p3);
 
      \node [transition] (t11) at(0.4,2)       {$a$}
       edge [pre]                             (p1)
       edge [post]                            (p3);

      \node [transition] (t2) at(1.1,2)       {$a$}
       edge [pre]                             (p1)
       edge [post]                            (p4);
 
      \node [transition] (t21) at(1.8,2)       {$a$}
       edge [pre]                             (p1)
       edge [post]                            (p4);

      \node [transition] (t7) at(3.13,2)       {$a$}
       edge [pre]                             (p2)
       edge [post]                            (p5);
 
      \node [transition] (t71) at(3.8,2)       {$a$}
       edge [pre]                             (p2)
       edge [post]                            (p5);

      \node [transition] (t8) at(4.5,2)       {$a$}
       edge [pre]                             (p2)
       edge [post]                            (p6);

      \node [transition] (t81) at(5.13,2)       {$a$}
       edge [pre]                             (p2)
       edge [post]                            (p6);

      \node [transition] (t3) at(-0.8,0)       {$b$}
       edge [pre]                 (p3);
%       edge [post,bend left]               (p1);

      \node [transition] (t4) at(2.2,0)       {$c$}
       edge [pre]                 (p4);
 
      \node [transition] (t5) at(2.8,0)       {$d$}
       edge [pre]                 (p5);
 
      \node [transition] (t6) at(5.5,0)       {$e$}
       edge [pre]                 (p6);
%       edge [post,bend right]               (p2);
 
      \draw[->,black,rounded corners] (t4)--(2.2,3)--(p1);  
      \draw[->,black,rounded corners] (t3)--(-0.8,3)--(p1);  
      \draw[->,black,rounded corners] (t5)--(2.8,3)--(p2);  
      \draw[->,black,rounded corners] (t6)--(5.5,3)--(p2);  
%      \draw[->,black,rounded corners] (t4)--(2,-0.5)--(-1,-0.5)--(-1,3)--(p1);  
%      \draw[->,black,rounded corners] (t5)--(3,-0.5)--(6,-0.5)--(6,3)--(p2);  
 
%\node[right=0cm,text width=6cm,font=\footnotesize] at (-1.5,-1.3){
%Figure($1$): T-system- live, 1-bdd, and S-component decomposable.
%Figure($2$): Trace-labelled FC-system having S-cover, not dp-representable
%};

%\node[right=0.5cm,text width=6cm,font=\footnotesize] at (3,2){
%Only S-component decomposition is:~
%$p_1ap_3 b p_1 \rmand p_1 a p_4 b p_1$.
%Language of Direct Product \emph(DP) formed by these
%components is Lang(DP,(1,2),(1,2))=\{ab\}, 
%while language of net is Lang(N,\{1,2,\},\{1,2\})= \{abb\}.
%};
%\begin{pgfonlayer}{background}
%       \filldraw [line width=4mm,join=round,black!7]
%       (-1,-0.7)  rectangle (2.3,3.5);
%\end{pgfonlayer}
%\begin{pgfonlayer}{background}
%       \filldraw [line width=4mm,join=round,green!7]
%       (2.9,-0.7)  rectangle (6.3,3.5);
%\end{pgfonlayer}
\end{scope}
%\label{non-dp-net} 
\end{tikzpicture}
 

%% file: ps.tex
%\section{Product systems}
%\label{lab:sec:ps}
%\input{ps}
 
We define product systems over a fixed distribution  
$(\Sigma_1,\Sigma_2,\ldots,\Sigma_k)$ of $\Sigma$.  
First we define sequential systems. 
\begin{definition}      
A \name{sequential system} over a set of actions $\Sigma_i$ is a finite state
automaton
$ A_i=\langle P_i,\to_i,G_i,p_i^0 \rangle$ where
$P_i$ are called \concept{states},
$G_i \subseteq P_i$ are final states,
$p_i^0 \in P_i$ is the initial state, and 
$ \to_i \subseteq P_i \times \Sigma_i \times P_i$ 
is a set of \concept{local moves}.
\end{definition} 
 
For a local move $t=\langle p,a,p' \rangle$ of $\rightarrow_i$ state 
$p$ is called \name{pre-state} sometimes denoted by \pre{t} and $p'$ is called 
\name{post-state} of $t$, sometimes denoted by \post{t}.
Such a move is sometimes called an \name{$a$-move} or an
\name{$a$-labelled} move.  
%This notation is extended to a set of local moves as well.  
 
Let $\to^i_a$ denote the set of all $a$-labelled moves in the sequential system $A_i$.
The \name{language} of a sequential system is defined as usual.

%\begin{definition} 
%Let $ A_i=\langle P_i,\to_i,G_i,p_i^0 \rangle$ be a sequential system over  
%alphabet $\Sigma_i$ for $1 \leq i \leq k$.
%A \concept{product system} $A$ over the distribution 
%$\Sigma=(\Sigma_1,\dots,\Sigma_k)$ is a tuple 
%$\langle A_1,\ldots,A_k \rangle$.
%\end{definition} 
% 
 
\begin{definition} 
Let $ A_i=\langle P_i,\to_i,G_i,p_i^0 \rangle$ be a sequential system over  
alphabet $\Sigma_i$ for $1 \leq i \leq k$.
A \concept{product system} $A$ over the distribution 
$\Sigma=(\Sigma_1,\dots,\Sigma_k)$  
sometimes denoted by $\langle A_1,\ldots,A_k \rangle$ is 
a tuple $\langle P, \To, R^0, G\rangle$, where :\\
\noindent
$P=\prodover_{i \in Loc} P_i$ is the set of \name{product states} of $A$;  
$R^0= (p_1^0,\dots,p_k^0)$ is the initial product state of $A$;
$G \subseteq \prodover_{i \in Loc}G_i$ is the set of final product states of $A$;  
and, $\To  \subseteq \bigcup_{a \in \Sigma} \To_a$, denotes the \name{global moves}  
of $A$ where  $\To_a = \prodover_{i \in loc(a)} \to^i_a$. 
\end{definition} 
Elements of $\To_a$ are sometimes called \name{global $a$-moves}. 
Any global $a$-move is \emph{global} within the set of component sequential machines  
where action $a$ occurs. For a global $a$-move $g$, we define its set of  
pre-states \name{pre($g$)} as the set of pre-states of all its component $a$-moves;  
the set of post-states \name{post($g$)} as the set of post-states of all its component $a$-moves;  
and, use notation $g[i]$ for its $i$-th component--local $a$-move--belonging to $A_i$,  
for all $i$ in $loc(a)$. We use $R[i]$ for the projection of a product state $R$ in $A_i$.
%and $R \project I$ for the projection to $I \subseteq Loc$.
 
%Let $\prodover_{i \in Loc} P_i$ be the set of \name{product states} of $A$.
%We use $R[i]$ for the projection of a product state $R$ in $A_i$,
%and $R \project I$ for the projection to $I \subseteq Loc$.
%The initial product state of $A$ is $R^0= (p_1^0,\dots,p_k^0)$.  
% 
%The \name{final states} of $A$ are denoted by $G$. 
 
\subsection{Direct products}

%Let $\To_a = \prodover_{i \in loc(a)} \to^i_a$ for the product system $A$. The set of  
%all \name{global moves} of $A$ is $\To =\bigcup_{a \in \Sigma} \To_a$. This move  
%is \emph{global} within the set of component sequential machines where label $a$ occurs. 
%Then for a global move $g$ in $\To_a$, we define its set of pre-states \name{pre($g$)}  
%as the set of pre-states of all its component $a$-moves; the set of post-states  
%\name{post($g$)} as the set of post-states of all its component $a$-moves; and,  
%use notation $g[i]$ for its $i$-th component--local $a$-move--belonging to $A_i$,  
%for all $i$ in $loc(a)$.
% 
%We define target-states global $a$-move $g$, 
%as $\post{g}= \cupover_{ i \in loc(a)} \post{g[i]}$ and
%source-states as 
%$\pre{g}= \cupover_{ i \in loc(a)} \pre{g[i]}$. 
 
With set of global moves $\To =\bigcup_{a \in \Sigma} \To_a$ and 
final states  $G= \prodover_{i \in Loc}G_i$ 
$A$ is called \name{product system with product-acceptance}.  
These systems are called \name{direct products} in \cite{Muk11}.
%(\pspac), 
 
With set of global moves $\To =\bigcup_{a \in \Sigma} \To_a$ and 
final states  $G \subseteq \prodover_{i \in Loc}G_i$ 
$A$ is called \name{product system with subset-acceptance}.
These systems are called \name{synchronous products} in \cite{Muk11}.
%(\pssac). 
 
The \name{runs} of a product system $A$ over some 
word $w$ are described by associating product states with prefixes of $w$: the  
empty word is assigned initial product state $R^0$, and for every  
prefix $va$ of $w$, if $R$ is the product state reached after $v$  
and $Q$ is reached after $va$ where,
for all $j \in loc(a), \langle R[j],a,Q[j] \rangle \in \to_j$,
and for all $j \notin loc(a), R[j]=Q[j]$. 
A run of a product system over word $w$ is said to be \name{accepting}  
if the product state reached after $w$ is in $G$. We define the  
\name{language} $Lang(A)$ of product system $A$, as the  set of 
words on which the product system has an accepting run.  
The set of languages accepted by direct (resp. synchronous)
products is called \name{direct (resp. synchronous) product languages}.  

We use a characterization from \cite{Muk11} of languages accepted by  
direct products. %, which is a given below. 
%\begin{proposition}
%\label{lm-ps-char}
%$L=\Lang(A)$ is the language of direct product system, 
% 
%$A=\langle A_1,\dots,A_k\rangle$
%defined over distributed alphabet $\Sigma$ iff 
%\[L = \{w \in \Sigma^* \mid \forall i \in \{1,\ldots,k\}, ~\exists u_i \in L 
%~\mbox{such that}~ w \project_{\Sigma_i} = u_i\project_{\Sigma_i} \}.
%\]
%Further $L=\Lang(A_1) \syncshuf \dots \syncshuf \Lang(A_k)$.
%\end{proposition}
%%
% 
% 
\begin{proposition}
\label{lm-ps-char}
Let $L$ be a language defined over distributed alphabet $\Sigma$. 
The language $L$ is a direct product language iff 
$L = \{w \in \Sigma^* \mid \forall i \in \{1,\ldots,k\}, ~\exists u_i \in L 
~\mbox{such that}~ w \project_{\Sigma_i} = u_i\project_{\Sigma_i} \}.
$
%Further $L=\Lang(A_1) \syncshuf \dots \syncshuf \Lang(A_k)$.
\end{proposition}
If $L=Lang(A)$ for direct product $A=\langle A_1,\dots,A_k\rangle$
defined over distributed alphabet $\Sigma$ then 
$L=\Lang(A_1) \syncshuf \dots \syncshuf \Lang(A_k)$.
 
We also use a characterization of synchronous product languages~\cite{Muk11}. 
\begin{proposition}%[\cite{Muk11}]
\label{lab:lm:spl}
    A language over distributed alphabet $\Sigma$ is accepted by a product system  
with subset-acceptance if and only if it can be expressed as a finite union of 
direct product languages. % with product acceptance condition. 
\end{proposition}  
 
The following property of direct products from \cite{PL14}  
clubs together the places of product system which correspond to places
of a cluster in the net. 
 
%\subsubsection{product systems with matchings}
%\begin{definition}[\psmatchings{} \cite{PL14}]
\begin{definition}[\psmatchings{}]
For global $a \in \Sigma$,  \name{\matching($a$)} is a subset of tuples
$\prodover_{i \in loc(a)} P_i$ such that  
for all $i$ in $loc(a)$,  projection of these tuples is the set of  
all pre-states of $a$-moves in $\rightarrow_i^a$,
and if a state $p \in P_i$ appears in one tuple,
it does not appear in another tuple.
We say a product state $R$ is \name{in \matching($a$)} 
if its projection $R \project loc(a)$ is in the matching.

A product system is said to have \name{matching of labels} if
for all global $a \in \Sigma$, there is a 
    suitable \matching($a$). Such a system is denoted by
    \name{\psmatchings{}}.
%A product system $A$ is said to have \concept{separation of labels} 
%if for all $i \in Loc$, and for all global actions $a$, if
%$\langle p,a,p'\rangle, \langle q,a,q'\rangle \in \to_i$ then $p=q$.
\label{cluster-ps-nd}
\end{definition}
We have \name{\psmatchings{} with product-acceptance},
if the set of final product states of it is a product of  
final states of component machines, or  
\concept{\psmatchings{} with subset-acceptance},
if the set of final product states is a subset of product of final states  
of individual components.

%Depending on the whether the set of final product states satisfy
%product condition we have  
%A \psmatchings{} with product acceptance condition (resp. subset acceptance condition)  
%is called \concept{\psmatchingspac{} (resp. \psmatchingssac)}.

%\begin{proposition}[\cite{PL15}]
%Let $A$ be an FC-matching product system.
%For any $i$,
%if there exist local moves $\langle p,a,p' \rangle, 
%\langle p,b,p'' \rangle ~\mbox{in}~ \to_i$,
%then $loc(a)=loc(b)$.
%\label{eq-loc}
%\end{proposition}  
%

% \begin{definition}[\cite{PL14}]
   A run of \psmatchings{} $A$ is said to be 
\name{consistent with a matching of labels}~\cite{PL14} if 
for all global actions $a$ and every prefix of the run
$R^0 \Step{v} R \Step{a} Q$,
the pre-states $R \project loc(a)$ are in the matching.
%\end{definition}

%Runs of a \psmatchings{} with consistent matchings, are defined in the same way as
%above where, an additional requirement of $\prodover_{j \in loc(a)} R[j] \in \matching(a)$, has to be satisfied. 

Consistency of matchings is a behavioural property and to check if a 
\psmatchings{} $A$ has it and can be done in \class{PSPACE}
\cite{PL14,PL15}. 
 
%\begin{proposition}[\cite{PL14,PL15}]        
%    \label{lab:pp:check-run-mat}
%    Given a product system with matchings, checking if all runs of $A$ are
%    consistent with matchings can be done in \class{PSPACE}.
%\end{proposition}  

%\subsubsection{Properties related free choiceness in product systems} 
The following property from~\cite{PL14} is used to capture free choice
property.
\begin{definition}[conflict-equivalent matchings for \psmatchings]
In a product system,
we say the local move $\langle p,a,q_1 \rangle \in \to_i$ is
\name{conflict-equivalent} to the local move
$\langle p',a,q'_1 \rangle \in \to_j$,
 if for every other local move
$\langle p,b,q_2 \rangle \in \to_i$, 
there is a local move $ \langle p',b,q'_2 \rangle  \in \to_j$  
and, conversely, for moves from $p'$ there are 
corresponding outgoing moves from $p$.
For global action a, its \matching($a$) is called 
\name{conflict-equivalent matching}, if
whenever $p,p'$ are related by the \matching($a$),
their outgoing local $a$-moves are conflict-equivalent.
\end{definition}
 
%A \psmatchingspac{} (resp. \psmatchingssac) is called 
%\concept{\fcpsmatchingspac{}} ( resp. \concept{\fcpsmatchingssac}) if 
%it has conflict-equivalent matching relations for all $a$. %\matching($a$) for all $a \in \Sigma$. 

%  
%\begin{figure}[!ht]
%\begin{center}
%%\frame{\input{./ce2ps-fig3}}
%\fbox{\input{./ps-fig}}
%\end{center}
%\caption{Product system $(A_1,A_2)$} 
%\label{ps-fig}
%\end{figure} 
% 
%Figure~\ref{ps-fig}, shows a product system  
%defined over a distributed alphabet $\Sigma=(\Sigma_1=\{a,b,c\},
%\Sigma_2=\{a,d,e\})$. It has two components $A_1$,and $A_2$ with  
%final states $F_1= \{r_1,r_2\}$ and, $F_2= \{s_1,s_2\}$, respectively. 
% 
% 
  
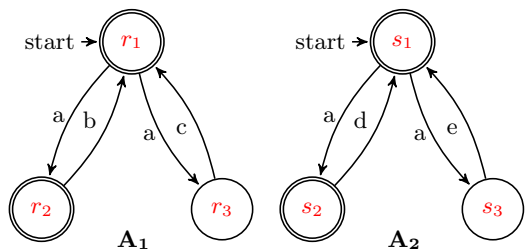
\begin{figure}[!htbp]
% \centering
	\begin{minipage}{.6\textwidth}
	%\frame{\input{./ps-fig}}
	\input{./ps-fig}
	%\fbox{\input{./ps-fig}}
	\caption{Product system $(A_1,A_2)$} 
	\label{ps-fig}
	\end{minipage}%
\vspace*{0.3cm}
    \begin{minipage}{0.4\textwidth}
%\begin{center}
Figure~\ref{ps-fig}, shows a product system  
defined over a distributed alphabet $\Sigma=(\Sigma_1=\{a,b,c\},
\Sigma_2=\{a,d,e\})$. It has two components $A_1$,and $A_2$ with  
final states $G_1= \{r_1,r_2\}$ and, $G_2= \{s_1,s_2\}$, respectively. 
%\end{center}
    \end{minipage}
\end{figure}

%\begin{example}[Product system with product-acceptance]
\begin{example}[Product system with matchings]
\label{lab-ex-dpmat}
Consider product system of Figure~\ref{ps-fig} and relation
$\matching(a)=\{(r_1,s_1)\}$ relation. 
This matching is conflict-equivalent and the system is consistent with
this matching relation. 
 
We have a \psmatchings{} $\calA=(A_1,A_2)$ with product
acceptance condition, if its set of final states is $G_1 \times G_2$.
%With $\globals(a)=\{$
%$((r_1 \step{a} r_2),(s_1 \step{a} s_2))$,
%$((r_1 \step{a} r_2),(s_1 \step{a} s_3))$,
%$((r_1 \step{a} r_3),(s_1 \step{a} s_2))$ 
%$((r_1 \step{a} r_3),(s_1 \step{a} s_3))$ 
%\}  
%relation, this system has same source property. And, it has product moves property. 
With the set of final states as
%$G = \{(r_1,s_1),(r_2,s_2)\}  \subseteq G_1 \times G_2$.
$\{(r_1,s_1),(r_2,s_2)\}  \subseteq G_1 \times G_2$, 
we have a \psmatchings{} $\calB=(A_1,A_2)$ having subset-acceptance. 
\end{example}  
%\begin{example}[Product systems with subset acceptance condition ]  
%\label{lab-ex-spmat}
 
%Consider product system of Figure~\ref{ps-fig}, with  
%$\calB=(A_1,A_2)$.
%$F_1= \{(r_1,r_2)\}$ and,
%$F_2= \{(s_1,s_2)\}$.  
%the set of final product states is 
%$G = \{(r_1,s_1),(r_2,s_2)\}  \subseteq G_1 \times G_2$.
%With $\matching(a)=\{(r_1,s_1)\}$ relation, we have a  
%\psmatchings{} $\calB=(A_1,A_2)$ having subset acceptance condition. 
%This matching is conflict-equivalent and the system is consistent with
%the matching relation. 
 
%With $\globals(a)=\{$
%$((r_1 \step{a} r_2),(s_1 \step{a} s_2))$,
%$((r_1 \step{a} r_2),(s_1 \step{a} s_3))$,
%$((r_1 \step{a} r_3),(s_1 \step{a} s_2))$ 
%$((r_1 \step{a} r_3),(s_1 \step{a} s_3))$ 
%\}   
%relation, this system has same source property. And, it has product moves property. 
%\end{example}  
% 

Lemma~\ref{lab:lm:sp-not-dp} presents a language not accepted by
any direct product. 
 
%Proof of Lemma~\ref{lab:lm:sp-not-dp} is given in
%Appendix~\ref{lab:app:b1}, which presents a language not accepted by
%any direct product. 
 
\begin{lemma}        
\label{lab:lm:sp-not-dp} 
The language $L_4= \{(abd+adb+abe+aeb+ace+aec+acd+adc)^*(\epsilon+a)\}$  
from Example~\ref{lab:ex:fcns-dcp-sac} 
%be the language defined over
%distributed alphabet $\Sigma=(\Sigma_1=\{a,b,c\}, \Sigma_2=\{a,d,e\})$.
%Then $L$ is not accepted by any direct product. 
is not accepted by any direct product. 
\end{lemma}  
\begin{proof}        
Consider a word $w=ab$ not in $L$ and,
words $u_1=abd$, $u_2=a$ which are in  $L$. 
We have projections, 
$w\project_{\Sigma_1}=ab=u_1\project_{\Sigma_1}$
$w\project_{\Sigma_2}=a=u_1\project_{\Sigma_2}$.
Therefore, by Proposition~\ref{lm-ps-char}, 
word $w$ is %$w \in \shuffle(L)$ and hence by
in $L$, which is a contradiction. 
\qed
\end{proof} 

We know that the class of synchronous product languages is strictly
larger than the class of direct product languages~\cite{Muk11}. 
With the matching relations this relationship is preserved.
The \psmatchings{} $\calB$ of Example~\ref{lab-ex-dpmat} accepts language $L_4$ 
%$\{(abd+adb+abe+aeb+ace+aec+acd+adc)^*(\epsilon+a)\}$ 
which by Lemma~\ref{lab:lm:sp-not-dp}, is not accepted by  
any direct product.  
%The \psmatchings{} $\calB$ accepts the language
%$\{(abd+adb+abe+aeb+ace+aec+acd+adc)^*(\epsilon+a)\}$ 
%which by Lemma~\ref{lab:lm:sp-not-dp}, is not accepted by  
%any direct product.  
Hence, the class of languages accepted by \psmatchings{} with subset-acceptance condition, 
is strictly larger, than the class of languages
accepted by \psmatchings{} with product-acceptance.  
 
However, using Proposition~\ref{lab:lm:spl} we have the following
characterization of \psmatchings{} with subset-acceptance. 
\begin{corollary}
\label{lab:cor:fcpsmatsac-lang-char}
A language $L$ is accepted by a product system with
subset-acceptance and,  having  conflict-equivalent and consistent matchings 
if and only if  
$L$ can be expressed as a finite union of languages accepted by
product system with product-acceptance  and, having
conflict-equivalent and consistent matchings.
\end{corollary}

Lemma~\ref{lab:lm:za-not-sp} presents a language not accepted by
any synchronous product.
 
%Proof of Lemma~\ref{lab:lm:za-not-sp} is given in
%Appendix~\ref{lab:app:b1}, which presents a language not accepted by
%any synchronous product.
 
\begin{lemma}        
\label{lab:lm:za-not-sp} 
The language $L_s= \{abd,adb,ace,aec\}^*(\epsilon+a)$ of Example~\ref{exfcnondcpnetpac} 
is not a synchronous product language. 
%Let $L= \{abd,adb,ace,aec\}^*(\epsilon+a)$ be a language defined over
%distributed alphabet $\Sigma=(\Sigma_1=\{a,b,c\}, \Sigma_2=\{a,d,e\})$.
%Then $L$ is not a synchronous product language. 
\end{lemma}  
\begin{proof}
%\begin{proof}[of Lemma~\ref{lab:lm:za-not-sp}]
If $L$ is accepted by any synchronized product then, $L$ can be expressed
as a finite union of direct product languages by Proposition~\ref{lab:lm:spl}.  
Let these direct product languages be $L_1,\ldots,L_k$.
Put $0$ for word $abd$ and $1$ for word $ace$, which are in $L$. 
Let $\calU =\{00\ldots0,10\ldots0,01\ldots0,\ldots,00\ldots1\}$ be the set
of $k+1$ words of length $k$ each. 
By pigeon hole principle, there must be two words of $\calU$ which
belong to same direct product language. Let $u$ and $v$ denote these two
words, and $L_i$ be the component language to which $u, v$ and belong to, where
$i \in \{1,\ldots,k\}$. 

Now we compare $u$ and $v$ to see how they are different from each other.
Either they differ in one position or in two different positions. 
\begin{enumerate}  
\item If $u$ and $v$ differ in only one position, then
$u=0^k$ i.e. $1$ does not occur in it, and 
      $v=0^m 1 0^{k-m}$ i.e. $1$ occurs at $m$-th position.  
      Now we consider word $w=(abd)^{m-1} (abe) (abd)^{k-m}$.
Clearly this word is not in $L$. We take projection of word $w$ %on $\Sigma$ to get  
      $w\project_{\Sigma_1}=(ab)^{m-1} (ab) (ab)^{k-m}=(ab)^k=
u\project_{\Sigma_1}$ and, 
      $w\project_{\Sigma_2}=(ad)^{m-1} (ae) (ad)^{k-m}=
v\project_{\Sigma_2}$. 
      By Proposition~\ref{lab:lm:spl}, the word $w$ is in $L_i$. And since
$L_i \subseteq L$, we have $w$ in $L$, which is a contradiction. 

\item If $u$ and $v$ differ in two positions, then $u$ has a $1$, and  
$v$ also has a $1$, but at a different position. Assume that $1$ of $u$
occurs at $m$-th position and $1$ in $v$ occurs at $m'$-th position. Without loss of
generality, we can assume that $m < m'$. Therefore $1 \leq m < m' \leq
k$. 
We consider word $w=(abd)^{m-1} acd (abd)^{m'-m-1} (abe)
(abd)^{k-m'}$, which is not in $L$. Now consider %projection
$w\project_{\Sigma_1}=(ab)^{m-1} (ac) (ab)^{m'-m-1} (ab) (ab)^{k-m'}= 
                     (ab)^{m-1} (ac) (ab)^{k-m} =
u\project_{\Sigma_1}$.

$w\project_{\Sigma_2}=(ad)^{m-1} (ad) (ad)^{m'-m-1} (ae) (ad)^{k-m'}= 
                     (ab)^{m'-1} (ae) (ad)^{k-m'} =
v\project_{\Sigma_2}$. 
      By Proposition~\ref{lab:lm:spl}, word $w \in L_i$ and, as
$L_i \subseteq L$  we have $w \in L$, which is a contradiction. 
\end{enumerate} 
\qed
\end{proof} 
% 
%We know that synchronous products do not accept all languages defined
%over distributed alphabets~\cite{Muk11}. 
So we have language $L_s$ which is not accepted by any \psmatchings{} with subset-acceptance. 
%We have proved that the language 
%$L= \{abd,adb,ace,aec\}^*(\epsilon+a)$ is not
%accepted by any \psmatchings{} with subset-acceptance. 
This motivates the bigger class of automata over distributed
alphabets, which we discuss next.

\subsection{Product systems with globals }
Let $A=\langle A_1,\ldots,A_k \rangle$ be a product system over distribution 
$\Sigma=(\Sigma_1,\dots,\Sigma_k)$ and, let  
\name{\globals($a$)} be a subset of its global moves $\To_a$, 
and \name{$a$-global} denote an element of \globals($a$).
 
\begin{definition}      
    A \name{product system with globals (PS-globals)} is a product system
with relations $\globals(a)$,  for each global action $a$ in $\Sigma$. 
 \end{definition}  
With subset-acceptance condition these systems are called Asynchronous
(or Zielonka) automaton~\cite{Zie87,Muk11}.
Runs of a product system with globals, are defined in the same way as
for the direct products, with an additional requirement of 
$\prodover_{j \in loc(a)}(\langle R[j],a,Q[j] \rangle \in
\globals(a)$, to be satisfied when $R \step{a}$ Q is to be taken. 
With abuse of notation sometimes we use $pre(a)$ to denote the set  
$\{R \mid \exists Q, R \step{a} Q\}$.

The following property from~\cite{P16}, of product systems with globals, relates
to free choice property of nets. 
\begin{definition}[same source property]       
A product system with globals have \name{same source property} if, any two  
global moves share a pre-state then their sets of pre-states are same.
\end{definition}  
 
%A language $L$ accepted by a \psmatchings{} (resp. \psglobals)
%is called a \concept{\psmatchings{}} (resp. \concept{\psglobals})
%\concept{language}. 
 %\begin{example}        
%Product system with globals and subset acceptance condition (Zielonka
%automata) $\calD=(A_1,A_2)$.
%$F_1= \{(r_1,r_2)\}$  and,
%$F_2= \{(s_1,s_2)\}$.  
% 
%The set of final product states is 
%$F = \{(r_1,s_1),(r_2,s_2)\}  \subseteq F_1 \times F_2$.
%
%The set of globals is $\globals(a)=\{$
%$((r_1 \step{a} r_2),(s_1 \step{a} s_2))$,
%$((r_1 \step{a} r_3),(s_1 \step{a} s_3)) \}$ 
%. 
% 
%This system has same source property. And it does not satisfy,
%product moves property. 
%\end{example}  
% 
%

\begin{example}[Product system with globals] %and product acceptance condition $\calC=(A_1,A_2)$.
\label{lab:ex:psglobals}
Consider the product system of Figure~\ref{ps-fig}. 
Let $\globals(a)=\{$
$((r_1 \step{a} r_2),(s_1 \step{a} s_2))$,
$((r_1 \step{a} r_3),(s_1 \step{a} s_3)) \}$. 
This system has same source property.  
%But it does not have product moves property. 
 
With the given \globals(a) we have a \psglobals{} $\calC=(A_1,A_2)$ with product
acceptance condition, if its set of final states is $G_1 \times
G_2$. And, for the set of final states 
$\{(r_1,s_1),(r_2,s_2)\}  \subseteq G_1 \times G_2$, 
we have a \psglobals{} $\calD=(A_1,A_2)$ with subset-acceptance. 
%$\psglobals{}$ $\calC$ and $\calD$ both have same source but not product moves property. 
\end{example}  
 
The language $L_s= \{abd,adb,ace,aec\}^*(\epsilon+a)\}$  
of Lemma~\ref{lab:lm:za-not-sp}, is 
accepted by product system with globals $\calD$ of
Example~\ref{lab:ex:psglobals} with same source property.

Product systems with globals and product-acceptance are not considered in~\cite{Muk11}.  
This class of systems are strictly more expressive, as shown in 
Lemma~\ref{lab:lm:zapac-not-dp}. %(proof in Appendix~\ref{lab:app:b1}).
This lemma is new and was not present in \cite{P18}. 
 
\begin{lemma}        
\label{lab:lm:zapac-not-dp} 
The language $L_p= \{(abd+adb+ace+aec)^*(\epsilon+a+ab+ad)\}$  
from Example~\ref{lab:ex:fcns-pac} 
is not accepted by any direct product. 
\end{lemma}  
\begin{proof}        
Consider a word $w=abe$ not in $L_p$ and,
words $u_1=abd$, $u_2=ace$ which are in $L_p$. 
We have projections, 
$w\project_{\Sigma_1}=ab=u_1\project_{\Sigma_1}$, 
$w\project_{\Sigma_2}=ae=u_2\project_{\Sigma_2}$.
Therefore, by Proposition~\ref{lm-ps-char}, 
word $w$ is %$w \in \shuffle(L)$ and hence by
in $L_p$, which is a contradiction. 
\qed
\end{proof}

%We give in Lemma~\ref{lab:lm:psglobalssac-lang-char}, a characterization of  
%class of languages accepted by product systems with globals and having  
%subset-acceptance, in terms of \psglobals{} and product-acceptance,  
%proof of which can be found in Appendix~\ref{lab:app:b1}.   
We give in Lemma~\ref{lab:lm:psglobalssac-lang-char}, a characterization of  
class of languages accepted by product systems with globals and having  
subset-acceptance, in terms of \psglobals{} and product-acceptance.

\begin{lemma}
    \label{lab:lm:psglobalssac-lang-char}
A language is accepted by a \psglobals{} with subset-acceptance if and
only if it can be expressed as a finite union of languages  
accepted by \psglobals{} with product-acceptance. 
\end{lemma}  
\begin{proof}%[of Lemma~\ref{lab:lm:psglobalssac-lang-char}]
    $(\Rightarrow)$: Let $A=\langle A_1,\ldots,A_k \rangle$ be a 
    \psglobals{}  with subset-acceptance condition, and
    having $(p^0_1,\ldots,p^0_k)$ as its initial state and set of
    final states $G \subseteq \prodover_{i \in Loc} G_i$, where 
    $A_i= \langle P_i, \to_i, G_i, p^0_i \rangle$. 
    Then for each final global state $g=(g_1,\ldots,g_k)$ of $G$,
    we build a \psglobals{} with product-acceptance condition
    $A^g= \langle A^g_1, \ldots, A^g_k \rangle$ by taking
    $A^g_i= \langle P_i, \to_i,g_i, p^0_i \rangle$.
    The set of globals of $A^g$ is the set of globals of $A$.  
    So if a word is accepted by $A$ by traversing a path from initial
    global state to some final state $g$, then we can traverse the same
    path in $A^g$ to its only one final global state $g$. And, the reverse  
    direction also holds. Therefore $Lang(A)=\cupover_{g \in G} Lang(A^g)$. 
 
\noindent
    $(\Leftarrow)$: Let $L=L^1 \cup \ldots \cup L^m$ be a language
defined over $\Sigma$ where each $L^j$ is the language accepted
by \psglobals{} $A^j$ with product-acceptance condition, 
for all $j$ in $\{1,\ldots,m\}$. 
Let $A^j=\langle A^j_1,\ldots,A^j_k\rangle$ where, its $i$-th
sequential-component over $\Sigma_i$ is  
$A^j_i= \langle P^j_i, \to^j_i, G^j_i, p^{j0}_i \rangle$, for all $i$ in $Loc$.  
 
Now we construct a product system with globals  
$B=\langle B_1,\ldots,B_k \rangle$ over $\Sigma$
having subset-acceptance condition as follows: 
    Each local component of $B$ over $\Sigma_i$ is given as 
    $B_i=\langle Q_i, \to^B_i, G^B_i, B^{0}_i \rangle$ 
    with local states $Q_i= \uplus_{j \in \{1,\ldots, m\}} P^j_i$  i.e. disjoint
    union of local states of $i$-th component of each \psglobals{} $A^j$. 
Its set of initial states is taken as union of initial states of $i$-th
	  components of \psglobals{} $A^j$ i.e.,  
	  $B^0_i= \cupover_{j \in \{1,\ldots, m\}}  p^{j0}_i$;
          and let its set of local moves be the union of local moves
	  of $i$-th component of \psglobals{} $A^j$; 
	  and its set of final states as the union of final states of
	  $i$-th component of \psglobals{} $A^j$. 
    Now we define final states $G^B \subseteq \prodover_{i \in Loc} G^B_i$ 
    of \psglobals{} $B$ as follows:
    $G^B= \cupover_{j \in \{1,\ldots,m\}} (G^j_1 \times \ldots \times
    G^j_k)$. This ensures that any word accepted in $L$ is accepted
    by $B$, and in the reverse direction, if we have a word $w$ accepted
by $B$, then we have an accepting run of some $A_j$ over $w$. 
\qed
\end{proof}  
%The construction given in the proof of Lemma~\ref{lab:lm:psglobalssac-lang-char},  
In the construction, %in the proof of Lemma~\ref{lab:lm:psglobalssac-lang-char},  
transition structure of local components is preserved, and so are the
global moves, hence, we have Corollary~\ref{lab:cor:fcpsglobalssac-lang-char}, 
which is used to get syntax for product systems with subset-acceptance condition. 
\begin{corollary}
    \label{lab:cor:fcpsglobalssac-lang-char}
    A language $L$ is accepted by a \psglobals{} with subset-acceptance and having same source property if and only if $L$ can be
    expressed as a finite union of languages accepted by \psglobals{}
with product-acceptance and same source property. 
\end{corollary}  

In a product system with globals and having same source property,
global moves for an action $a$ can be partitioned into different
compartments : two global $a$-moves belong to same compartment if they
have the same set of pre-states. 
%Let $\To_a^1,  \To_a^2, \ldots, \To_a^l$ be the compartments of these partition.
For any $a$-global $g$ of a same source compartment
$\To_a^{SS}$, %of $a$-global moves, 
we associate a \name{target-configuration}
%$\pi(g)= \prodover_{i \in loc(a)} (\targetstates(g) \cap P_i)$.  
$\pi(g)= \prodover_{i \in loc(a)}\post{g} \cap P_i $.  
Let $\post{\To_a^{SS}} = \{\pi(g) \mid g \in \To_a^{SS}\}$. 
%Let $\post{\To_a^{SS}} = \cupover_{g \in \To_a^{SS}} \pi(g)$. 
We define  $\To_a^{SS}[i]=\post{\To_a^{SS}} \cap P_i$ and 
$\postdecomp(\To_a^{SS})= \prodover_{i \in loc(a) } \To_a^{SS}[i]$. 
%We define post-projection of a compartment as $\To_a^{SS}[i]=\post{
%\To_a^{SS}} \cap P_i$ and, post-decomposition as  
%$\postdecomp(\To_a^{SS})= \prodover_{i \in loc(a) } \To_a^{SS}[i]$. 
We may call  $\To_a^{SS}[i]$ as \name{post-projection} and  
$\postdecomp(\To_a^{SS})$ as \name{post-decompoistion} of a compartment.

The following property relates to distributed choice property of nets.  
 
\begin{definition}      
\label{lab:def:pmp}
A product system with globals and having same source property,
is said to have \concept{product moves property}, if for all $a$ in
$\Sigma$, and
for all same source compartments $\To_a^{SS}$ of $a$-globals, 
$\postdecomp(\To_a^{SS}) \subseteq \post{\To_a^{SS}}$. 
\end{definition}  
 
Product systems $\calC$ and $\calD$ of Example~\ref{lab:ex:psglobals} do not have 
product moves property.  
%Product system of Example~\ref{lab:ex:pmp} in Appendix~\ref{lab:app:exps}, have product moves property. 
Product system of Example~\ref{lab:ex:pmp} has product moves property. 
 
\begin{example}[Product system with globals and product moves property] 
\label{lab:ex:pmp}
Consider product system $\calA$ of Example~\ref{lab-ex-dpmat}
where the set of final states is $G_1 \times G_2$.
With $\globals(a)=\{$
$((r_1 \step{a} r_2),(s_1 \step{a} s_2))$,
$((r_1 \step{a} r_2),(s_1 \step{a} s_3))$,
$((r_1 \step{a} r_3),(s_1 \step{a} s_2))$ 
$((r_1 \step{a} r_3),(s_1 \step{a} s_3))$ 
\}  
relation, we have \psglobals{} $\calA'$ which has product moves property. 
This also has same source property.

Now consider the product system $\calB$ of Example~\ref{lab-ex-dpmat},
with $\{(r_1,s_1),(r_2,s_2)\} \subseteq G_1 \times G_2$ as its final
states, and having the $\globals(a)$ relation as above, we get a  
\psglobals{} $\calB'$ with subset-acceptance condition, having 
product moves and same source property. 
\end{example}

%% 
%\begin{example}        
%Product system with globals and subset acceptance condition $\calB'=(A_1,A_2)$.
%$F_1= \{(r_1,r_2)\}$ and,
%$F_2= \{(s_1,s_2)\}$. The set of final product states is 
%$F = \{(r_1,s_1),(r_2,s_2)\}  \subseteq F_1 \times F_2$.
% 
%With $\globals(a)=\{$
%$((r_1 \step{a} r_2),(s_1 \step{a} s_2))$,
%$((r_1 \step{a} r_2),(s_1 \step{a} s_3))$,
%$((r_1 \step{a} r_3),(s_1 \step{a} s_2))$ 
%$((r_1 \step{a} r_3),(s_1 \step{a} s_3))$ 
%\}   
%relation, this system has same source property. And, it has product moves property. 
%\end{example}  
% 
% 
%
 
A product system with globals is said to be \name{live}, if for any
global move $g$ and any reachable product state $R$, there exists a
product state $Q$ such that $g$ is enabled at $Q$.

\subsection{Relating product systems with matchings and globals
}
\label{lab:sec:pmp2mat}

First we show, in Theorem~\ref{lab:tm:pmp2mat}, how to construct a product system with  
consistent and conflict-equivalent matchings from a \psglobals{} with same source
property.  
 
\begin{theorem}      
\label{lab:tm:pmp2mat}
Let $\Sigma$ be a distributed alphabet and $A$ be a product system with  
globals defined over it.%$\Sigma$. 
Then we can construct a product system $B$ with matchings, linear in the size of
product system $A$ with globals such that, 
\begin{enumerate}  
\item if $A$ has same source property then $B$ has conflict-equivalent
matchings, 
\item in addition, if $A$ is live then 
\begin{enumerate}  
\item $B$ is consistent with matchings, and
\item $\Lang(A)=\Lang(B)$. 
\end{enumerate} 
\end{enumerate}  
\end{theorem}  
 
\begin{proof}%[of Theorem~\ref{lab:tm:pmp2mat}]
Let $A=(A_1,\ldots,A_k)$ be a product system with globals. 
We construct product system $B=(B_1=A_1,\ldots,B_k=A_k)$ with
matchings, where for each label $a$ in $\Sigma$, 
$\matching(a)=\{\prodover_{i \in loc(a)} \pre{g[i]} \mid g \in
\globals(a)\}$. 
The size of this $\matching(a)$ relation is at most the
size of $\globals(a)$ relation for $e$.

\begin{enumerate}  
\item For a label $a$ in $\Sigma$, without loss of generality, we assume
that $loc(a)=\{1,2\}$. Let $(p,q)$ is in $\matching(a)$. Let 
$\langle p,a,p_1 \rangle$ be a local $a$-move and 
$\langle p,b,p_2 \rangle$ be a local $b$-move in $B_1$.
Also let $\langle q,a,q_1 \rangle$ be a local $a$-move of $B_2$.
To prove conflict-equivalence of matchings, we have to show existence
of a local $b$-move from state $q$ of $B_2$. 
Since, $(p,q)$ is in $matching(a)$ then we have a global $a$-move 
$g=( \langle p,a,p_1 \rangle, \langle q,a,q_1 \rangle)$ of $A$.  
And, since $p$ has an outgoing $b$-move there must be some other
global $b$-move $g'$ in $A$, with $g'[1]= \langle p,b,p_2 \rangle$.  
But, we have  $\{p\} \subseteq \pre{g'} \cap \pre{g}$.  
Therefore, by same source property of $A$, 
we get $\pre{g'} = \pre{g}$.
Hence, $q$ is in $\pre{g'}$ implying existence of a local
$b$-move from state $q$, as required.

\item Now in addition, we assume that $A$ is live. 
\begin{enumerate}  
\item Consider a run $R_0 \step{v} R \step{a}Q$ of product system $B$,
where $R_0$ is a initial global state of $B$. 
Inductively, we assume that the run $R_0 \step{v}R$ is 
consistent with the constructed matchings.  
Without loss of generality, we assume that $loc(a)=\{1,2\}$.  
At global state $R$ some global $a$-move 
$g=(\langle p,a,p_1 \rangle, \langle q,a,q_1 \rangle, \ldots \langle
r_n,a,r'_n \rangle)$ of $B$ is enabled to reach $Q$. Other component
moves of $g$ from $A_3$ to $A_4$ do not take part in thie step.  
To show that this run is consistent with matchings,  
we have to prove, tuple $R\project_{loc(a)}=(p,q)$ is in $\matching(a)$. 
If global move of $g$ of $B$ is an $a$-global in system $A$, then
clearly by construction $(p,q)$ being in $\pre{g}$ also appears in
$\matching(a)$.     
If global move $g$ of $B$ is not an $a$-global in system $A$, then  
we must be having $a$-globals $g$ and $g'$ of $A$, such that
$g[1]=\langle p,a,p_1 \rangle$, and 
$g'[2]=\langle q,a,q_1 \rangle$.  
At global state $R$ of $A$, local state of $A_1$ is $p$ and
local state of $A_2$ is $q$. 
%There is no other $a$-global whose first two component moves are
%same as that of $g$ is enabled at $R$, 
Because of same source property of $A$, there is no global move on any other
label having $p$ and $q$ in their preset, which is enabled i.e., in these two
components control will not be able to move forward from $p$ and $q$,
which contradicts the fact that $A$ is live. 
 
\item We show language equivalence of $A$ and $B$ by showing a stronger
property that graphs of reachable states of $A$ and $B$ are isomorphic. 
States of $A$ are mapped to themselves in $B$ and vice versa. 
As base step, initial states of $A$ and $B$ are isomorphic. 

To prove that $Lang(A) \subseteq Lang(B)$, 
we show that if at any reachable state $R$ of $A$,
some $a$-global $g$ is taken to reach $Q$ then  
there exists a global move on label $a$ in $B$ which is enabled at $R$ and
when taken we reach global state $Q$. 
Since $g$ is enabled at $R$ in system $A$, 
for all $i$ in $loc(a)$, we have  $g[i]= \langle p_i,a,q_i \rangle$ where 
$p_i$ a locate state of $A_i$ is in $\pre{g}$ and part of $R$ and
similarly, $q_i$ a local state of $A_i$ is in $\post{g}$ and part of $Q$. 
Therefore, we can take $a$-global $g$ of $A$ itself as the required
global move of $B$ taken at $R$ to reach $Q$. 
 
In the reverse direction, % of proving $Lang(B) \subseteq Lang(A)$, 
we assume that we have reached state $R$ in system $B$, after taking
an $a$-labelled global move $h$ to reach state $Q$. 
Inductively, we assume that we are at state $R$ in $A$. 
Now we have to show that there exists an $a$-global $g$ in $A$, which is
enabled at $R$ and $R \step{g} Q$ in $A$.  
Let $loc(a)=\{1,\ldots,m\}$. 
We have $\pre{h}$ appearing in $R$ is also in $\matching(a)$ of $B$,
due to consistency of matchings for $B$, proved above using liveness of $A$. 
Therefore, by construction of $B$, we must have some 
$a$-global $g$ in $A$ such that $\pre(g)=\pre{h}$.
Let $\To_a^{SS}$ be the set of global moves of $A$ which have same set of
pre-states as $g$.  
For all $i$ in $loc(a)$, $h[i]=\langle p_i,a,q_i \rangle$ in each component
$B_i$ of $B$ and hence in each $A_i$ of $A$. 
Therefore, we must have global moves $g_1,g_2,\ldots, g_m$ in
$\To_a^{SS}$ such that $g_i[i]=h[i]$, for all $i$ in $loc(a)$. 
Hence $\prodover_{i \in loc(a)} q_i$ is in $\postdecomp(\To_a^{SS})$ in
$A$.
This tuple $\prodover_{i \in loc(a)} q_i$ is also in $\post(\To_a^{SS})$
because $A$ has product moves property.
So there exists a global move $g'$ in $\To_a^{SS}$ with $\pi(g')=
\prodover_{i \in loc(a)} q_i$. 
Therefore, this global $g'$ in $A$ can be fired at $R$ to reach $Q$,
as required, to complete induction. 
\end{enumerate}   
\end{enumerate}  
\qed 
\end{proof}  

Now, from a \psmatchings{} with consistent and
conflict-equivalent matchings we construct a product system with
globals having same source property. 
\begin{theorem}
\label{lab:tm:mat2pmp}
Let $\Sigma$ be a distributed alphabet and let $B$ be a product system
with conflict equivalent and consistent matchings. 
Then for the language of $B$ we can construct a product system $A$ with globals  
over $\Sigma$ having same source and product moves property. The 
constructed product system $A$ with globals is exponential in the size of
system $B$ having matching of labels. 
\end{theorem}  
 
\begin{proof}%[Proof of Theorem~\ref{lab:tm:mat2pmp}]        
Let $B=(B_1,B_2,\ldots,B_n)$ be the given product system with matchings. 
We construct a \psglobals{} $A=(A_1,A_2,\ldots,A_n)$ by 
taking component systems $A_j=B_j$.
It remains to construct the global moves of $A$. 
We build the set of global moves $\globals(a)$ for $A$ from a given matching relation
$matching(a)$ of $B$ as follows. 
Let $loc(a)=\{1,\ldots,m\}$.  
\[\begin{array}{l}
        \globals(a)= \{( (p_1,\ldots,p_m),(q_1,\ldots,q_m)) \mid 
                (p_1,\ldots,p_m) \in \matching(a) \\
	\hspace*{6.4cm} ~\mbox{where}~(p_k \to_k^a q_k) 
	~\mbox{for all}~ k~\in loc(a) \}.%\{1,\ldots,m\}%~\rmand~i \in loc(a). 
%\rmand ~\mbox{importantly}~ (p_1,\ldots,p_m)~\\
%	\hspace*{0.51cm} \mbox{is part of some reachable marking of}~$A$. 
    \end{array}
\]
 
For tuple $(p_1,\ldots, p_m)$ in $\matching(a)$ of $B$, let us assume that
$p_1$ has $k_1$ outgoing local $a$-moves in $B_1$,
$p_2$ has $k_2$ outgoing local $a$-moves in $B_2$,
and so on. Let $k$ be the minimum of $\{k_1,\ldots,k_m\}$. 
Then we have $k^m$ number of global $a$-moves, which is exponential in
the number of locations. 
 
%\noindent 
(\emph{Proof of $A$ having same source property}):\\
Assume not,then we have two global moves $g$ and $g'$  
such that intersection of their sets of pre-states is not empty and
their sets of pre-states are not equal also. 
Let $g$ is an $a$-global move and $g'$ be a $b$-global move of $A$.  
For the sake of simplifying this discussion, we take $m=2$, hence
$loc(a)=\{1,2\}$. We have $\{p,q\} \subseteq \pre{g}$ and 
$\{p,q'\} \subseteq \pre{g'}$, and $q \neq q'$, where states $p$ is 
a local state of $A_1$ and $q,q'$ are local states of $A_2$.  
Therefore, $g = ((p,a, p_1),(q, a, q_1))$ and $g' = ((p, b, p_2),(q', b, q_2))$. 
It means that $(p,q') \in \matching(b)$ and $(p,q) \in \matching(a)$ in $B$.  
Since $B$ has conflict-equivalent matching, it implies existence of a local $b$-move  
with source state $q$  i.e., $(q,b,q_3)$ in $B_2$, for some state
$q_3$ in $B_2$. 
 
Now at any global state $R$ of $B$ where $g$ is enabled, tuple
$(p,q)$ which is in $\matching(a)$ relation, also appears in state $R$,  
Therefore, the global $b$-move  $h= ((p,b, p_2),(q, b, q_3))$ is also enabled in $B$. 
But, $B$ has consistent matching of labels, so $(p,q) \in \matching(b)$.  
This is a contradiction, as now we have two tuples $(p,q)$ and $(p,q')$ in
$\matching(b)$ relation in which state $p$ appears.  
Therefore, $q=q'$ to get a contradiction. 
%and our assumption of $B$
%not having same source property must be wrong. 
 
(\emph{Proof of $A$ having product moves property}):\\
Let  $\To_a^{SS}$ be a same source compartment of global $a$-moves of
the constructed product system $A$ and 
$loc(a)=\{1,\ldots,m\}$. 
Now we have to prove 
$\postdecomp(\To_a^{ss}) \subseteq \post(\To_a^{SS})$. 
Consider a tuple of states $(p_1,p_2,\ldots,p_m)$ in $\postdecomp(\To_a^{ss})$.
Then there exist global moves $g_1,  g_2, \ldots, g_m$  in
$\To_a^{SS}$ such that 
$\post{g_{i}[i]}= p_i$, for all $i$ in $loc(a)$, implying that we have 
local $a$-moves $g_{j}[j]=\langle q_j,a,p_j \rangle$ in $A_j$, for all
$j$ in $loc(a)$. 
Since %$A$ have same source property as proved above, and  
$(q_1,\ldots,q_m)= \prodover_{j \in loc(a)} \pre{g_{j}[j]}$, 
we have tuple $(q_1,\ldots,q_m)$ in $\matching(a)$ of product
system $B$. Therefore, we have $((\prodover_{j \in loc(a)} q_j),
(\prodover_{j \in loc(a)} p_j))$ as a global $a$-move in the
constructed system $A$. 
Hence, tuple $\prodover_{j \in loc(a)} p_j))$ being a post-configuration of
this global $a$-move is also in $\post(\To_a^{SS})$, as required. 
 
(\emph{Proof of $\Lang(B)=\Lang(A)$}):\\
We prove this by showing isomorphism of state reachability graphs of
$A$ and $B$. For states it is identity mapping. 
Proving $Lang(A) \subseteq \Lang(B)$ is straightforward, as global
moves of $A$ are also global moves of $B$. 

In the reverse direction, % of proving $Lang(B) \subseteq \Lang(A)$,  
we assume that we have some global move $h$ of $B$ taken at reachable state $R$,  
to reach $Q$. Inductively, we have reached global state $R$ in $A$.  
We have $\pre{h}$ in $R$ and $\post{h}$ in target state $Q$.
As B is consistent with matching of labels, we have $\pre{h}$ in
$\matching(a)$. Since $A$ has product moves property, proved above using  
consistency of matching, for each tuple of $\matching(a)$, we have an
$a$-global, for this fixed pre-states and each possible post-configuration.  
So we have an $a$-global in $A$, consisting of $\post{h}$ as its
set of post-states and \pre{h} as its set of pre-states.  
We can take this at $R$ to reach $Q$ in $A$ as required.  
\qed
\end{proof}  
  
%Proofs of Theorem~\ref{lab:tm:mat2pmp} and Theorem~\ref{lab:tm:pmp2mat} 
%can be found in Appendix~\ref{lab:app:b1}. 

%% file: ps-fig.tex
%%% product systems with globals--figure
\begin{tikzpicture}[->,>=stealth',shorten >=1pt,%
auto,node distance=2.8cm,semithick,inner sep=2pt,bend
angle=15,scale=.6]

  \tikzstyle{every state}=[fill=white!20,ellipse,draw=black!,text=red]
%\begin{tikzpicture}[->,>=stealth',shorten >=1pt,auto,node distance=2.8cm, inner sep=2pt]
%\draw[help lines] (0,0) grid(3,2);

\begin{scope}

   %\node[state,fill=cyan!20]             (A) at(-2,0)  {$r_2$};
   \node[state,accepting]             (A) at(-2,0)  {$r_2$};
   \node[state]                        (B) at(2,0)   {$r_3$};
   %\node[initial,state,above,accepting,fill=yellow!20]       (E) at(0,3)  {$r_1$}; 
   \node[initial,state,above,accepting]       (E) at(0,3)  {$r_1$}; 
 
   \node[state,accepting]             (C) at(4,0)  {$s_2$};
   %\node[state,fill=cyan!20]             (C) at(4,0)  {$s_2$};
   \node[state]                        (D) at(8,0)   {$s_3$};
   \node[initial,state,above,accepting]       (F) at(6,3)  {$s_1$}; 
   %\node[initial,state,above,accepting,fill=yellow!20]       (F) at(6,3)  {$s_1$}; 

   \path (A) edge[bend right]  node {b} (E)
         (B) edge[bend right]  node {c} (E)
         %(E) edge[bend right]    node [left] {\textcolor{red}{a}} (A)
         (E) edge[bend right]    node [left] {a} (A)
             edge [bend right]   node [left] {a} (B);

   \path (C) edge[bend right]  node {d} (F)
         (D) edge[bend right]  node {e} (F)
         %(F) edge[bend right]    node [left] {\textcolor{red}{a}} (C)
         (F) edge[bend right]    node [left] {a} (C)
             edge [bend right]   node [left] {a} (D);

%       \node[right=1.5cm,text width=4.5cm,font=\footnotesize] at (6.5,2.6){
%$\globals(a)=\{$\\
%$((r_1 \step{a} r_2),(s_1 \step{a} s_2))$,\\
%$((r_1 \step{a} r_3),(s_1 \step{a} s_3)) \}$. 
%};
%  
      \node[below=0.8cm,text width=1cm] at (0.5,1.0){
            $\mathbf{A_1}$
        };

      \node[below=0.8cm,text width=1cm] at (6.5,1.0){
            $\mathbf{A_2}$
        };
%\begin{pgfonlayer}{background}
%       \filldraw [line width=4mm,join=round,black!7]
%       (-1,-0.7)  rectangle (6,3.5);
%\end{pgfonlayer}
%
%\draw[->,black,rounded corners](3)--(0,-0.5)--(1.5,-0.5)node[above]{$a$} --(3.3,-0.5) --(3.3,5.8);         
\end{scope}
\end{tikzpicture}

%\end{document}

%% file: n2p.tex
%\section{Nets and Product systems}
%\label{sec-nets}

We first present a generic construction of  
a $1$-bounded S-decomposable 
%a $1$-bounded S-coverable
labelled net systems, from product systems with globals. %

%g\begin{definition}[\psglobals{} to nets]
%g\label{def:generic}
%gGiven a \psglobals{} $A=\langle A_1,A_2,\ldots,A_k \rangle$ 
%gover distribution $\Sigma$, a net system  
%g$(N=(S,T,F,\lambda),M_0,\calG)$ is constructed as follows:
%g\begin{itemize}  
%g\item $S=\bigcup_i P_i$, the set of places.
%g\item $T=\bigcup_a \globals(a)$, for all actions $a$ in $\Sigma$.
%gWe also let $T_i = \{\lambda^{-1}(a) \mid a \in \Sigma_i\}$.
%g\item The labelling function $\lambda$ labels by action $a$ the
%gtransitions in $\globals(a)$.
%g\item The flow relation 
%g$F= \{(p,g),(g,q)\mid g \in T_a, g[i] = \langle p,a,q \rangle, i \in loc(a) \}$.
%gLet $F_i$ be its restriction to the transitions $T_i$ for $i \in loc(a)$.
%g\item $M_0= \{p^0_1,\dots,p^0_k\}$, the initial product state.
%g\item $\calG = G$, the set of final global states.
%g\end{itemize} 
%g\end{definition}
%g 
 
\begin{definition}[\psglobals{} to nets]
\label{def:generic}
Given a \psglobals{} $A=\langle A_1,\ldots,A_k \rangle$ 
over distribution $\Sigma$, a net system  
$(N=(S,T,F,\lambda),M_0,\calG)$ is constructed as follows:
%\begin{itemize}  
The set of places is $S=\bigcup_i P_i$, 
%the set of transitions $T=\bigcup_a \globals(a)$,  
the set of transitions is $T=\bigcup_{a \in\Sigma} \globals(a)$.
%for all actions $a$ in $\Sigma$. 
Define $T_i = \{\lambda^{-1}(a) \mid a \in \Sigma_i\}$.
The labelling function $\lambda$ labels by action $a$ the
transitions in $\globals(a)$.
The flow relation is
$F= \{(p,g),(g,q)\mid g \in T_a, g[i] = \langle p,a,q \rangle, \mbox{~for all~} i \in
loc(a) \}$, define $F_i$ as its restriction to the transitions $T_i$ for $i \in loc(a)$.
See that $F$ is union of all $F_i$s. 
Let $M_0= \{p^0_1,\dots,p^0_k\}$, be the initial product state and 
$\calG = G$ as the set of final global states.
%\end{itemize} 
\end{definition}

We get one to one correspondence between 
reachable states of product system and reachable markings of nets
because the set of transitions of resultant net is same as the set of
global moves in the product system, and construction preserves pre as well as
post places.  

\begin{lemma}
The constructed net system $N$ from a \psglobals{} $A$,  
as in Definition~\ref{def:generic},
is S-decomposable and $Lang(N,M_0,\calG)=\Lang(A)$.
%is S-coverable and $Lang(N,M_0,\calG)=\Lang(A)$.
The size of constructed net is linear in the size of product system. 
\end{lemma}

Applying the generic construction above to product systems with
same source property, we get a free choice net,
because any two global moves having same set of 
pre-places are put into one cluster. 

\begin{theorem}
\label{lab:tm:pmp2dcp}
    Let $(N,M_0,\calG)$ be the net system constructed from
    \psglobals{} $A$ as in Definition~\ref{def:generic}. 
\begin{itemize}  
\item If $A$ has same source property %(i.e. $A$ is an \fcpsglobalssac{})  
      then $N$ is a free choice net,
\item  In addition if $A$ has product moves property, then $N$ has
distributed choice.  
\end{itemize} 
%\label{lab:p2n:ext}
\end{theorem}  

In the construction, if the product system has subset-acceptance then we get a net  
with a set of final markings, which may not have product condition. 
Since $A$ has subset-acceptance, we generalize the results obtained in \cite{P16}. 
 
For the product system $\calD$ of Example~\ref{lab:ex:psglobals},
accepting language $L_s$ we can construct 
the net system of Example~\ref{exfcnondcpnetpac}.
 
In the case that the product system $A$ has matchings, 
transitions of net constructed are reachable global moves of system $A$~\cite{PL14,PL15}.  
%And if had subset acceptance condition then set of final markings of net need not have product acceptance condition.  
 
%\subsection{pmp to fc-dcp and back}
 
%\begin{theorem}      
%\label{lab:tm:pmp2dcp}
%Let $\Sigma$ be a distributed alphabet and $A$ be a product system with  
%globals having same source and product moves property, defined over $\Sigma$.  
%the net constructed $(N,M_0,\calG)$ 
%and having distributed choice property.  
%%The size of constructed net is linear in the size of product system. 
%\end{theorem}  
% 
%
%\begin{theorem}      
%\label{lab:tm:dcp2pmp}
%Let $\Sigma$ be a distributed alphabet and let $(N,M_0,\calG)$ be a free choice net
%system over $\Sigma$, and having distributed choice property.  
%Then we can construct a product system $A$ with globals over $\Sigma$ for the  
%language of $N$, having same source and product moves property. The size of product  
%system $A$ is linear in the size of net. 
%\end{theorem}  
% 
% 

%\subsection{S-coverable net systems to product systems}
%\label{sec-n2p}
 
%Even if a net is $1$-bounded and S-coverable each component
%need not have only one token in it, but when we say that a $1$-bounded net
%is S-coverable we assume that each component has one token.
%For live and $1$-bounded free choice nets, such $S$-covers 
%can be guaranteed \cite{DesEsp95}.
%Now we describe a linear-size construction of a product system
%from a net which is S-coverable.

Now we describe a linear-size construction of a product system
from a net which is S-decomposable.
%\begin{definition}[Net systems to product system with globals and
%    subset acceptance condition (\psglobalssac{}]
%\begin{definition}[nets to product system with globals]
%\label{def:direct}
%Given a $1$-bounded labelled and S-coverable net system $(N,M_0,\calG)$, with $N=(S,T,F,\lambda)$ 
%the underlying net and $N_i=(S_i, T_i, F_i)$ the components in the S-cover, 
%for $i$ in $\{1,2,\ldots,k\}$, we define a product system 
%$A=\langle A_1,\ldots,A_k \rangle$.
%\begin{itemize}
%\item $P_i= S_i$, 
%$p_i^0$ the unique state in $M_0 \cap P_i$.
%\item 
%$\to_i= \{\langle p,\lambda(t),p' \rangle \mid t \in T_i ~\mbox{and}~
%(p,t), (t,p') \in F_i, ~\mbox{for}~p,p' \in P_i
%\}$.
%\item
%So we get sequential systems $A_i= \langle P_i,\to_i,p_i^0 \rangle$
%and the product system $A=\langle A_1,A_2,\ldots,A_k \rangle$
%over alphabet $\Sigma$.
%\item $\globals(a) = \cupover_{t \in T_a} (\prodover_{i \in loc(a)} t[i])$.
%\item $G = \calG$ 
%\end{itemize}
%\end{definition}
 
 \begin{definition}[nets to \psglobals]
\label{def:direct}
Given a $1$-bounded labelled and an S-decomposable net system $(N,M_0,\calG)$, with $N=(S,T,F,\lambda)$ 
%Given a $1$-bounded labelled and an S-coverable net system $(N,M_0,\calG)$, with $N=(S,T,F,\lambda)$ 
the underlying net and $N_i=(S_i, T_i, F_i)$ the components in the S-cover, 
for $i$ in $\{1,2,\ldots,k\}$, we define a product system $A=\langle
A_1,\ldots,A_k \rangle$, 
as follows. 
Take $P_i= S_i$, and $p_i^0$ the unique state in $M_0 \cap P_i$.
Define local moves $\to_i= \{\langle p,\lambda(t),p' \rangle \mid t \in T_i ~\mbox{and}~
(p,t), (t,p') \in F_i, ~\mbox{for}~p,p' \in P_i \}$.
So we get sequential systems $A_i= \langle P_i,\to_i,p_i^0 \rangle$, 
and the product system $A=\langle A_1,A_2,\ldots,A_k \rangle$
over alphabet $\Sigma$.
Global moves are $\globals(a) = \{\prodover_{i \in loc(a)} t[i] \mid t
\in T_a \} $.
And, the set of final states is
$G = \calG$. 
\end{definition}

\begin{lemma}
From net system $N$ with a final set of markings, the construction 
of the \psglobals{} $A$ in Definition~\ref{def:direct} above preserves language.
The product system $A$ is linear in the size of net, and
product system has subset-acceptance. 
\end{lemma}
For each $a$-labelled transition of the net we get one global $a$-move in the  
product system having same set of pre-places and post-places.   
And, for each global $a$-move in product system we have an $a$-labelled transition  
in the net having same pre and post-places.
%This was ensured by ``distributed choice'' property \cite{PL14,Pha-thesis},
%here now ensured by construction using globals.
We get one to one correspondence between reachable states of product system and  
reachable markings of the net we started with. Therefore, if we begin with a  
free choice net, we get same source property in the obtained product system.   
And, for each transition in the net we have a global trasition hence, $A$ has  
product moves property if the net has distributed choice. 

\begin{theorem}
\label{lab:tm:dcp2pmp}
%\label{tm-dc-n2p} 
Let $(N,M_0,\calG)$ be a $1$-bounded, and an  
S-decomposable 
%S-coverable 
labelled net with a set of final markings $\calG$. %free choice net system. %(FCNS-SAC) 
Then 
\begin{itemize}  
\item  if $N$ has free choice property, then  constructed product system $A$ 
with globals, has same source property,  
\item  in addition, if the net has distributed choice, then 
       $A$ has product moves. 
\end{itemize} 
\end{theorem} 
In construction of Definition~\ref{def:direct}, 
we start with a net having distributed choice and a final set of markings then  
we get product system with
matching with subset-acceptance condition. Note that in this case we do not have
to construct globals~\cite{PL14,PL15}.
 
For the net system of Example~\ref{exfcnondcpnetpac} 
and accepting language $L_s$ we can construct 
the product system $\calD$ of Example~\ref{lab:ex:psglobals}.
Therefore, we generalize the results from~\cite{PL14,PL15}.  
\begin{theorem}
%For a $1$-bounded, S-coverable labelled net having distributed choice and  
For a $1$-bounded, S-decomposable labelled net having distributed choice and  
given with a set of final markings. 
Then one can construct a product system with 
conflict-equivalent and consistent matchings and 
having subset-acceptance. 
\label{tm-dc-n2p2} 
\end{theorem}

Given below is the converse result. 
\begin{theorem} 
\label{lab:tm:ps2n-dcp1}
For a product system with conflict-equivalent, consistent
matchings, and subset-acceptance, we get language equivalent free choice net
with distributed choice and having a set of final markings. 
\end{theorem}

%% file: expr.tex
%\section{Expressions} 
%\label{lab:sec:expr} 
 
%Fix a distribution $(\Sigma_1,\ldots,\Sigma_k)$ of $\Sigma$. 
First we define regular expressions and its %Antimirov~\cite{Mir66}
derivatives. 

\subsection{Regular expressions and their properties}
\label{lab-sec-dfn}
\noindent
A \name{regular expression} over alphabet $\Sigma_i$ 
such that constants $0$ and $1$ are not in $\Sigma_i$ is given by:
\[
s ::= 0 \mid 1 \mid a \in \Sigma_i \mid s_1 \cdot s_2 \mid s_1 + s_2 \mid s_1^*
\]
The language of constant $0$ is $\emptyset$ and 
that of $1$ is $\{\epsilon\}$.  
For a symbol $a \in \Sigma_i$, its language is $\Lang(a)=\{a\}$.  
For regular expressions $s_1+s_2, s_1 \cdot s_2 ~\mbox{and}~
s^*_1$, its languages are defined inductively
as union, concatenation and Kleene star 
of the component languages respectively.

As a measure of the size of an expression we will use
$wd(s)$ for its \name{alphabetic width}---the total number of occurrences of
letters of $\Sigma$ in $s$.  
 
For each regular expression $s$ over $\Sigma_i$,  
let $\Lang(s)$ be its language and  
its \name{initial actions} 
form the set $\Init(s) = \{a\mid \exists v\in \Sigma_i^{*}~\mbox{and}~av \in \Lang(s)~\}$
which can be defined syntactically.
We can syntactically check whether the empty word
$\epsilon \in \Lang(s)$.  
 
We use derivatives of regular expressions
which are known since the time of Brzozowski \cite{Brz64},
Mirkin \cite{Mir66} and Antimirov \cite{Ant96}.
 
\begin{definition}[Antimirov derivatives~\cite{Ant96}]
Given regular expression $s$ and symbol $a$, the set of partial \name{derivatives} 
of $s$ with respect to $a$, written $\Der_{a}(s)$ are defined as follows.

\begin{math}
\begin{array}{rcl}
\Der_{a}(0) & = & \emptyset	\\
\Der_{a}(1) & = & \emptyset	\\
\Der_{a}(b) & = & \{\epsilon\} \If~b=a,~~~\emptyset~\mbox{otherwise}\\
%\Der_{a}(a) & = & \{\epsilon\}	\\
\Der_{a}(s_1+s_2) & = & \Der_{a}(s_1)\cup \Der_{a}(s_2) \\
\Der_{a}(s_1^{*}) & = & \Der_{a}(s_1) \cdot s_1^{*}\\
\Der_{a}(s_1 \cdot s_2) & = & \left \{
 \begin{array}{l l}
    \Der_{a}(s_1) \cdot s_2 \cup \Der_{a}(s_2), 
		& \If \epsilon \in \Lang(s_1)\\
    \Der_{a}(s_1) \cdot s_2 & \Otherwise
 \end{array}  \right. 
\end{array}
\end{math}
\end{definition}
Inductively $\Der_{aw}(s) = \Der_w(\Der_a(s))$.\\
The set of all partial derivatives $\Der(s) = \cupover_{w \in \Sigma_i^*} \Der_w(s)$,
where $\Der_{\epsilon}(s)=\{s\}$.
We have derivatives $\Der_{a}(ab+ac)=\{b,c\}$ and
$\Der_{a}(a(b+c))=\{b+c\}$.
 
A derivative $d$ of $s$ with action $a \in \Init(d)$ is 
called an \name{$a$-site} of $s$. 
An expression is said to have \name{equal choice} 
if for all $a$, its $a$-sites have the same set of initial actions.
For a set $D$ of derivatives, we collect all initial actions to form $\Init(D)$.
Two sets of derivatives have equal choice if their \Init~sets are same.
 
As in \cite{PL14} we put together derivatives which may correspond to the same state 
in a finite automaton.

\begin{definition}[\cite{PL14}]
Let $s$ be a regular expression and $L=\Lang(s)$.
For a set $D$ of $a$-sites of regular expression $s$ and an action $a$, 
we define the relativized language 
$L^D_a = \{xay \mid xay \in L, \exists d \in \Der_{x}(s) \cap D,
\exists d' \in  \Der_{ay}(d)~\mbox{with}~ \epsilon \in \Lang(d')\}$,
and the prefixes $\Pref^D_a(L) = \{x \mid xay \in L^D_a \}$,
and the suffixes $\Suf^D_a(L) = \{y \mid xay \in L^D_a \}$.
We say that the derivatives in set $D$ \name{$a$-bifurcate} $L$
if $L^D_a = \Pref^D_a(L)~a~\Suf^D_a(L)$.
%If $D$ is the set of all $a$-sites, then we say $s$ is \concept{$a$-bifurcated}. 
\label{lab:defi:bif}
\end{definition}
 
We use \name{partition}s of the $a$-sites of $s$ into blocks such that 
each \name{block} (that is, element of the partition) $a$-bifurcates $L$~\cite{PL14}.  
%We have given this definition in Appendix~\ref{lab:app:b}.  

\begin{definition}[\cite{PL14}]
\label{def-part}
Let $X_1$ be a partition of $a$-sites of $s_1$ and $X_2$ be a partition of
$a$-sites of $s_2$, where regular expression $s=s_1 \cdot s_2$ or $s=s_1 + s_2$.
For partitions $X_1,X_2$ with blocks $D_1,D_2$
containing elements $d_1,d_2$ respectively, we use the notation
$(X_1 \cup X_2)[d/d_1,d_2]$ for the modified partition
$((X_1 \setminus \{D_1\}) \cup (X_2 \setminus \{D_2\}) \cup 
\{(D_1 \cup D_2 \cup \{d\}) \setminus \{d_1,d_2\}\}$. 
And, for partition $X$ with block $D_1$ in it,
having $d_1$ in it,
$X[d/d_1]$ is the modified partition 
$X \setminus \{ D_1 \} \cup \{(D_1 \setminus \{d_1\})\cup \{d \} \}$.

\begin{center} 
\begin{math}
\begin{array}{rcl}
\Part_a(b) & = & \emptyset \If a \neq b \\
\Part_a(a) & = & \{\{a\}\} \\
\Part_a(s_1^*) & = & (\Part_a(s_1) \cdot s_1^*)[s_1^*/s_1 \cdot s_1^*]\\
\Part_a(s_1+s_2) & = &  Z_1 \cup Z_2 \cup \{s_1 + s_2 \} \If a \in \Init(s_1 + s_2)\\
\Part_a(s_1 \cdot s_2) & = & \left \{
\begin{array}{l l}
 	\Part_a(s_1)\cdot s_2 \cup \Part_a(s_2)[s_1 \cdot s_2/s_2]
		& \If \epsilon \in \Lang(s_1) \\
                & \rmand  \epsilon \notin \Lang(s_2)\\
       	\Part_a(s_1) \cdot s_2 \cup \Part_{a}(s_2) & \Otherwise
\end{array}  \right. 
\end{array} 
\end{math}
\Where,\\
\begin{math}
\begin{array}{l l}
Z_1 =\Part_a(s_1)\setminus\{s_1\} & \If s_1 \notin \Der_a(s_1 + s_2),~\Part_a(s_1) \Otherwise\\
Z_2=\Part_a(s_2)\setminus\{s_2\}&  \If s_2 \notin \Der_a(s_1 +
s_2),~\Part_a(s_2) \Otherwise.
\end{array} 
\end{math}
\end{center}
\end{definition}

For an action $a$, let $\Part_a(s)$ denote such a partition. 
In addition to thinking of 
blocks of the partition as places of an automaton, we can think of 
pairs of blocks and their effects as local moves. 
%blocks and their effects as local moves. 
%The proof follows that in the thesis \cite{Pha-thesis}.
%We will put this idea to work in our key Lemma~\ref{lem-uniq-to-sep} later.

\begin{definition}[\cite{P16}]
Given an action $a$, and a set of $a$-sites $B$ of regular expression
$s$, and a specified set of \name{$a$-effects} $E \subseteq Der_{a}(B)$,
we define the relativized languages 

\[\begin{array}{l}
%L^D_a = \{xay \mid xay \in L, \exists d \in Der_{x}(s) \cap D, \exists d' \in  Der_{ay}(d)~\mbox{with}~ \epsilon \in Lang(d')\},\\
L^{(B,E)}_a = \{xay \in L\mid \exists d \in Der_{x}(s) \cap B,
\exists d' \in  Der_{a}(d) \cap E,~\mbox{and}~ \\
 \hspace*{6.4cm} \exists d'' \in  Der_{y}(d')~ \mbox{with}~ \epsilon \in Lang(d'')\}.
%\hspace*{8cm}\mbox{with}~ \epsilon \in Lang(d'')\}.
\end{array}
\]

We define the prefixes $Pref^{(B,E)}_a(L) = \{x \mid xay \in L^{(B,E)}_a \}$
and the suffixes $Suf^{(B,E)}_a(L) = \{y \mid xay \in L^{(B,E)}_a \}$. 
%$Suf^{B,E}_a(L) = \{y \mid xay \in L^{(D,E)}_a \}$.
We say that %the set $D$ \concept{$a$-bifurcates} $L$ if $L^D_a = Pref^D_a(L)~a~Suf^D_a(L)$ and
a tuple $(B,E)$ \name{$a$-funnels} $L$ 
if $L^{(B,E)}_a = Pref^{B}_a(L)\cdot a\cdot Suf^{(B,E)}_a(L)$.
In such a pair $(B,E)$, if $B$ is a block in the $\Part_a(s)$ 
and $E$ is a nonempty subset of $a$-effects of $B$,
then it is called as an \name{$a$-duct}. 
\label{lab:defi:bif-ext}
\end{definition}
For an $a$-duct $(B,E)$, we define its set of initial actions  
$\Init(B,E)$ as $\Init(B,E)=\Init(B)$,
call $B$ as its \name{pre-block} %of $D$,
and call $E$ as its \name{post-effect}. %of $D$.
For all $i$ in $loc(a)$ let \name{$a$-ducts}($s_i$) denote the set
of all $a$-ducts of regular expression $s_i$.  
For any two $a$-ducts $(B,E)$ and $(B',E')$ in $a$-ducts($s_i$),
define $(B,E)=(B',E')$ if $B=B'$ and $E=E'$. 
Given an $a$-duct $d=(B,E)$ its post-effect $E$ 
is sometimes denoted by $\postset{d}$ and its pre-block $B$ can be denoted
as $\preset{d}$. For a collection of ducts $z$, the set of all their  
post-effects (resp. pre-blocks) is denoted as $\postset{z}$ (resp. $\preset{z}$). 
In a similar way, we define the set of post-effects of an  $a$-cable $D$, 
%as $\postset{D}= \cupover_{ i \in loc(a)} \postset{D[i]}$ and
as $\postset{D}= \{\postset{D[i]} \mid i \in loc(a)\}$ and
its set of pre-blocks as $\preset{D}= \{ \preset{D[i]}\mid i \in loc(a)\}$. 
 
\subsection{Connected expressions over a distributed alphabets}
\label{lab:section-exp-def}
 
The syntax of \name{connected expressions} defined over 
a distribution $(\Sigma_1,\Sigma_2,\ldots,\Sigma_k)$ of alphabet
$\Sigma$ is given below.
%The regular expression $s_i$ can be of any star-height,
%which is different from our earlier paper \cite{LMP11}.
\[
e ::= \zero |\fsync(s_1,s_2,\ldots,s_k), \mbox{~where~}s_i \mbox{~is a regular
expression over~} \Sigma_i
\]
When $e=\fsync(s_1,s_2,\ldots,s_k)$ and $I \subseteq Loc$,
let the \name{projection} $e \project I = \prodover_{i \in I} s_i$.

%\begin{definition}      
A connected expression $e=\fsync(s_1,s_2,\dots,s_k)$ over $\Sigma$,
is said to have \name{equal choice} if, 
for all global actions $a$ in $\Sigma$
and for any $i$, $j$ in $loc(a)$,
any $a$-site of $s_i$ have same \Init~set as of any $a$-site of $s_j$.  
 
% A connected expression $e=\fsync(s_1,s_2,\dots,s_k)$ over $\Sigma$,
%is said to have \concept{unique sites} if, 
%each component regular expression $s_i$ have unique sites property. 
%\end{definition}  
 
For a connected expression defined over distributed alphabet
its derivatives and semantics were given in \cite{PL14}, and
are given as follows.  
For the connected expression $\zero$, we have 
$Lang(\zero)=\emptyset$.
For the connected expression $e=\fsync(s_1,\ldots,s_k)$,
its language is 
%\vspace{-0.3cm}
%\[Lang(e)=Lang(s_1) \syncshuf Lang(s_2)\syncshuf \ldots \syncshuf Lang(s_k),
%\]
$Lang(e)=Lang(s_1) \syncshuf Lang(s_2)\syncshuf \ldots \syncshuf
Lang(s_k)$. 
 
%\vspace{-0.3cm}
%where the \name{synchronized shuffle} $L = L_1 \syncshuf \dots \syncshuf L_k$ 
%is defined by 
%\vspace{-0.3cm}
%\[w \in L \Iff \Forall i \in \{1,\ldots,k\}, w\project_{\Sigma_i} \in L_i.
%\]
The definitions of derivatives
extended to connected expressions \cite{PL14} is as follows. 
The expression $\zero$ has no derivatives on any action.
Given an expression $e=\fsync(s_1,s_2,\ldots,s_k)$, its derivatives 
are defined by induction using the derivatives of the $s_i$ on action $a$:

%\vspace{-0.3cm}
\begin{center}
$Der_a(e)=\{\fsync(r_1,\ldots,r_k) \mid \forall i \in loc(a), 
r_i \in Der_a(s_i); \Otherwise~r_j=s_j\}.
$
\end{center}

\subsection{Connected expressions with pairings}
\label{lab:subsec:prop-ce}
We recall some properties of connected expressions over a
distribution, which were, useful in construction of free choice nets.  
This property relates to matchings of direct products~\cite{PL14}. 
%We produce a new property of connected expressions which is \class{Ptime} checkable, 
%and useful in giving alternate syntax for free choice nets with 
%distributed choice property. 
 
%All but the last property are \class{Ptime}-checkable.
%The last property is checkable in \class{Pspace} since it runs over 
%all reachable derivatives.
 
%\subsubsection{Connected expressions with pairing (\cepairing)}
\begin{definition}[\cite{PL14}]
Let $e=\fsync(s_1,s_2,\dots,s_k)$ be a connected expression over $\Sigma$.
For a global action $a$, \name{\pairing(a)} is a subset of tuples
%$\prodover_{i \in loc(a)} Part_a(s_i)$, the projections of these
%tuples covering the $a$-sites in $s_i$, such that if a block
$\prodover_{i \in loc(a)} \Part_a(s_i)$ such that the projection of these
tuples includes all the blocks of $\Part_a(s_i)$, and if a block
of $\Part_a(s_j), j \in loc(a)$ appears in one tuple of the pairing,
it does not appear in another tuple.
(For convenience we also write $\pairing(a)$ as a subset of 
$\prodover_{i \in loc(a)} \Der(s_i)$ which respects the partition.)
We call $\pairing(a)$ \name{equal choice} if for every tuple
in the pairing, the blocks of derivatives in the tuple have equal choice.
\end{definition}
%\noindent
Derivatives for connected expressions with \pairing{} are defined as
follows. 
A derivative $\fsync(r_1,\dots,r_k)$ is \name{in \pairing(a)} if 
there is a tuple $D \in \pairing(a)$ 
such that $r_i \in D[i]$ for all $i \in loc(a)$.
For convenience we may write a derivative as an element of $\pairing(a)$.
Expression $e$ is said to have \name{(equal choice) pairing of actions} 
if for all global actions $a$, there exists an (equal choice) pairing of $a$.
Expression $e$ is said to be \name{consistent with a pairing of actions}
if every reachable $a$-site $d \in \Der(e)$ is in $\pairing(a)$.
Expression $e$ is said to have \name{equal choice} property if
it has equal choice pairing of actions for all global actions $a$ in $\Sigma$.

%\begin{proposition}[\cite{PL14,PL15}]
%    \label{lab:pp:conruncheck}
    Given a connected expression $e$ with pairings, checking if it is
    consistent with pairing of actions can be done in
\class{PSPACE}~\cite{PL14,PL15}.
%\end{proposition}  

\subsection{Connected expression with cables (\cecables)}
We give some properties of connected expressions over a
distribution, which extend the notion of pairing,  
and have been related to product systems with globals \cite{P16}. 
The notion of cables corresponds to notion of globals of product
systems, and hence it corresponds to transitions of a net.

\begin{definition}[\cite{P16}]
Let $e=\fsync(s_1,s_2,\dots,s_k)$ be a connected expression over $\Sigma$.
For each action $a$ in $\Sigma$, we define
\name{$a$-cables(e)} = $\prodover_{ i \in loc(a)}$$a$-$ducts(s_i)$.  
For an action $a$, an \name{$a$-\cable} is an element of the set $a$-$cables(e)$.  
We say that a block $B$ of $\Part_a(s_i)$ \name{appears} in an
$a$-\cable $D$ if there exists $j$ in $loc(a)$ and there exists $Y \subseteq
\Der_a(B)$ such that $D[j]=(B,Y)$, i.e. if $B$ is a pre-block of a
component $a$-duct of $D$. 
For any $a$-\cable  $D$, its set of pre-blocks $^\bullet D
= \cup_{i \in loc(a)} \{ B_i \mid B_i ~\mbox{appears in}~ D\}$, i.e.
the set of pre-blocks of all the of its component $a$-ducts. 

\noindent
For expression $e$, let \name{\cables(a)} $\subseteq
a$-$cables(e)$, such that for all $i$ in $loc(a)$  
\begin{enumerate}  
%\item for all $(B,E)$ in $a$-$ducts(s_i)$, there exists $D$ in
%$cables(a)$, such that there exists $j$ in $loc(a)$, such that
%$D[j]=(B,E)$, i.e. each $a$-duct~of $s_i$ appears in at least one $a$-\cable of it.  
\item Each block $B$ in $\Part_a(s_i)$, appears in at least one  
$a$-\cable of it.  
\item for all $(B,E)$ and $(B',E')$ in $a$-$ducts(s_i)$ with $(B,E) \neq
(B',E')$, if $B=B' \implies E \cap E' = \emptyset$, i.e.
if any two distinct $a$-ducts of $s_i$ appearing in it have same pre-block then,  
they must have disjoint post-effects.
\end{enumerate}  
 
Connected expressions with cables were defined in \cite{P16}, as follows. 

\noindent 
A \name{connected expression with cables (\cecables{})} is a connected expression 
with relations $\cables(a)$ of it,   
for each global action $a$ in $\Sigma$. 
\end{definition} 

Derivatives of a connected expression with cables are \cite{P16}
defined as follows.
%\begin{definition}[\cite{P16}] 
The \cecables{} $\zero$ has no derivatives on any action.
For expression $e=\fsync(s_1,s_2,\ldots,s_k)$,
we define its derivatives on action $a$,
by induction, using $a$-ducts and the derivatives of $s_j$ as:

\begin{center} 
$\Der_a(e)=\{\fsync(r_1,r_2,\ldots,r_k) \mid ~r_j \in \Der_a(s_j)$
if there exists an $a$-\cable  $D$ in $\cables(a)$
such that, for all $j$ in $loc(a)$,
$s_j$ is in pre-block $B_j$ and $r_j$ is in $X_j$ of 
$a$-duct $D[j]=(B_j,X_j)$ of $s_j$,
$
 \Otherwise~r_j=s_j\}.
$
\end{center}
%\end{definition}
  
We use the word \name{derivative} for expressions such as 
$d=\fsync(r_1,\dots,r_k)$ given above. The \name{reachable} derivatives are
$Der(e) = \{d \mid d \in Der_x(e), x \in \Sigma^*\}$.
A \cecables{} is said to have \name{equal source property} if for any pair of
two \cables~sharing a common pre-block have same set of pre-blocks.
This property corresponds to same source property of product systems
and relates to transitions belonging to same cluster of nets.

Language of $e$ is the set of words over $\Sigma$ defined using derivatives as below. 
\begin{center}
$Lang(e)=\{w \in \Sigma^*\mid \exists e' \in \Der_w(e) ~\mbox{such that}
~\epsilon~ \in Lang(r_i),~\mbox{where}~ e'[i]=r_i\}$.
\end{center}
 
So we can have next derivative on action $a$, if it is allowed by  
the $\cables(a)$ relation.  
The number of derivatives may be exponential in $k$. 
Let $\Sigma=(\Sigma_1=\{ a,b,c\},\Sigma_2=\{ a,d,e\})$ be a distributed alphabet.  
\begin{example}[\cepairings{} and \cecables] %%%% connected  expr \fsync((ab+ac),(ad+ae))^*
\label{ex-ce1pairing} 
%Let $\Sigma=(\Sigma_1=\{ a,b,c\},\Sigma_2=\{ a,d,e\})$ be a distributed alphabet.  
Let $e=\fsync((ab+ac)^*, (ad+ae)^*)$ be a connected expression defined over $\Sigma$. 
Here, $r_1=(ab+ac)^*$ %is a regular expression defined over component alphabet $\Sigma_1$ and, $s_1=(ad+ae)^*$ is defined over $\Sigma_2$. 
and $s_1=(ad+ae)^*$. % is defined over $\Sigma_2$. 
The set of derivatives of $r_1$ is $\Der(r_1)= \{ r_1, r_2=b r_1,
r_3=c r_1 \}$ and for  $s_1$ it is $\Der(s_1)= \{ s_1, s_2=d s_1, s_3=e s_1 \}$.  
We have $a$-sites$(r_1)=r_1$ and $\Part_a(r_1) = \{D=r_1 \}$. 
Similarly, $a$-sites$(s_1)=r_1$ and $\Part_a(s_1) = \{D'=s_1\}$. 
The only possible pairing relation is $pairing(a)=\{(D,D')\}$. 
We have $\Der_a(e)= \{\fsync(r_i,s_j) \mid i,j \in \{2,3\}\}$. 
Expression $e$ satisfies equal choice property. 

Now we associate a cable relation with $e$. 
The set of $a$-effects of $D$ is $\Der_a(D)= \{ r_2 , r_3 \}$.  
The set of $a$-ducts of $r_1$ is $\{(D,r_2), (D,r_3),(D,\{r_2,r_3\}) \}$. 
The set of $a$-effects of $D'$ is $\Der_a(D')= \{ s_2 , s_3 \}$ and 
the set of $a$-ducts of component expression $s_1$ is $\{(D',s_2), (D',s_3),(D',\{s_2,s_3\})\}$. 
A possible $\cables(a)$ relations for expression $e$ is
$\{((D, r_2), (D',s_2)), ((D, r_3), (D',s_3))\}$. 
See that each block in the $\Part_a(r_1)$ and 
$\Part_a(s_1)$ appears at least once in the $\cables(a)$ relation. 
And two $a$-ducts of $r_1$ appearing in this relation, have 
same pre-block $D$, so their set of post-effects ${r_2}$  
and ${r_3}$ are disjoint. This condition also holds for $a$-ducts
of $s_1$.

For both $a$-cables set of pre-blocks is identical, 
therefore $\cables(a)$ satisfies equal source property. 
We have $\Der_a(e)= \{\fsync(r_2,s_2), \fsync(r_3,s_3)\}$, 
but expression $\fsync(r_2,s_3)$ is not in $\Der_a(e)$, because
only post-effect of $D$ containing $r_2$ is the set $\{r_2\}$ and 
similarly, only post-effect of $D'$ containing $s_3$ is the set $\{s_3\}$ and 
there does not exist an $a$-cable with $(D,\{r_2\})$ and $(D',
\{s_3\})$ as its components. 
%The set of derivatives is
We have $\Der(e)=\{e, (r_2,s_2),(r_3,s_3),(r_1,s_2),(r_1,s_3),(r_2,s_1),(r_3,s_1) \}$. 
\end{example} 
 
%Another such example %of connected expression with pairings (resp. cables)  
%can be found %is given %below whose details can be found in Appendix~\ref{lab:app:d}. 
%in Appendix~\ref{lab:app:d}. 
Another such example of connected expression with pairings (resp. cables)  
is given below.
 
\begin{example}[\cepairings{}] %%%% \fsync((ab+ac)*a,(ad+ae)*a)^*
\label{ex-ce2pairing} 
%Let $\Sigma=(\Sigma_1=\{ a,b,c\},\Sigma_2=\{ a,d,e\})$ be a  distributed alphabet.  
Let $e=\fsync((ab+ac)^*a, (ad+ae)^*a)$ be a  
connected expression defined over $\Sigma$. 
Let $p_1=(ab+ac)^*a$ %is a regular expression defined over component alphabet $\Sigma_1$,  
with language $L_1$% = \Lang(r_1)$ and,
and $q_1=(ad+ae)^*a$ %is defined over $\Sigma_2$  
with language $L_2$. %= \Lang(s_1)$.  
The set of derivatives are  
$\Der(p_1)= \{p_1, p_2=b p_1, p_3=c p_1 , p_4=\epsilon\}$ and 
$\Der(q_1)= \{ q_1, q_2=d q_1, q_3=e q_1 , q_4=\epsilon\}$.  
The partitions of $a$-sites are $\Part_a(p_1) = \{B=\{p_1\} \}$ and 
%Its set of $a$-sites consist of $p_1$ only.  
$\Part_a(q_1) = \{B'=\{q_1\} \}$. 
A pairing relation is $pairing(a)=\{(B,B')\}$ and 
%With respect to the $pairing(a)$ relation,  
with respect to that $\Der_a(e)= \{\fsync(p_i,q_j) \mid i,j \in \{2,3,4\}\}$. 
Expression $e$ has equal choice property. 
%The set of derivatives of expression $e$ is
%$\Der(e)=\{e, (r_2,s_2),(r_3,s_3),(r_1,s_2),(r_1,s_3),(r_2,s_1),(r_3,s_1) \}$. 
%\end{example} 
 
Now we associate a cabling relation with $e$.  
%\begin{example}[\cecables{}] %%%% connected  expr \fsync((ab+ac),(ad+ae))^*
%\label{ex-ce2cables} 
%Consider connected expression of Example~\ref{ex-ce2pairing}. 
The set of $a$-effects of $B$ is  $\Der_a(B)= \{ p_2 , p_3,p_4 \}$
and  $\Der_a(B')= \{ q_2 , q_3,q_4 \}$.  
The set of $a$-ducts for $p_1$ is  $\{(B,\{p_2\}),
(B,\{p_3\}),(B,\{p_2,p_4\}), (B,\{p_3,p_4\}) (B,\{p_2,p_4,p_3\}) \}$
and for $q_1$ is  $\{(B',q_2),
(B',q_3),(B',\{q_2,q_4\}), (B',\{q_3,q_4\}) (B',\{q_2,q_4,q_3\}) \}$. 
A \cables(a) relation is 
$\{((B, p_2), (B',q_2)), ((B, p_3), (B',q_3)),((B,
p_2), (B',q_2))\}$. 
Expression $e$ has equal source property and its set of derivatives
with respect to letter $a$ is 
$\Der_a(e)= \{\fsync(p_2,q_2), \fsync(p_3,q_3),\fsync(p_4,q_4)\}$. 
%The set of derivatives of expression $e$ is
%$\Der(e)=\{e, (r_2,s_2),(r_3,s_3),(r_1,s_2),(r_1,s_3),(r_2,s_1),(r_3,s_1) \}$. 
\end{example}

Now we give two new properties of connected expressions. 
%\subsubsection{Product derivatives property:}
In a connected expression with cables and having equal source property,
cables for an global action a can be partitioned into different
compartments : two $a$-cables belong to same compartment if they
have equal source. 
%Let $\ES_a^1,  \ES_a^2, \ldots, \ES_a^l$ be the compartments of these partition.
To any such $a$-cable $D$ belonging to an equal source compartment $\ES_a$ of $a$-cables, 
we can associate a set of its \name{post-blocks} listed in some order as 
$\pi(D)= \prodover_{i \in loc(a)} (\postset{D} \cap \chi_i)$, 
where $\chi_i = \{ \Part_a(s_i)\mid a \in \Sigma \} $.
Let $\post{\ES_a} = \{\pi(D) \mid D \in \ES_a \} $. 
%We define post-projection of compartment $\ES_a$ as $\ES_a[i]= \post{\ES_a ^i} \cap \chi_i$,  
%and its post-decomposition as 
%$\postdecomp(\ES_a)= \prodover_{i \in loc(a) } \ES_a[i]$. 
Then we define $\ES_a[i]= \post{\ES_a ^i} \cap \chi_i$
and $\postdecomp(\ES_a)= \prodover_{i \in loc(a) } \ES_a[i]$. 
We call $\ES_a[i]$ as \name{post-projection} and  
$\postdecomp(\ES_a)$ as \name{post-decomposition} of the compartment $\ES_a$. 
 
The following property of connected expressions will later be related to  
product-moves property of direct
products and hence to distributed-choice of nets.  
\begin{definition}[product-derivatives property]
\label{lab:def:pdp}
A connected expression with cables and having equal source property,
is said to have \concept{product-derivatives property}, if for all $a$  
in $\Sigma$, and for all equal source compartments $\ES_a$ of
$a$-cables $\postdecomp(\ES_a) \subseteq \post{\ES_a}$. 
\end{definition}  
  
\begin{example}        
The connected expression $e=\fsync((ab+ac)^*, (ad+ae)^*)$ 
of Example~\ref{ex-ce1pairing}
does not have product-derivative property with the given cabling
relation $\{((D, r_2), (D',s_2)), ((D, r_3), (D',s_3))\}$. 
If we associate the cabling relation $\{((D, r_2), (D',s_2)), ((D, r_3), (D',s_3),
((D, r_2), (D',s_3)), ((D, r_3), (D',s_2))\}$ with $e$ 
then it has product-derivative property.
\end{example}

%\fbox{\parbox{\linewidth}{\textcolor{red}{\textsf{ 
%Put this definition of action-live ness of expressions in prelimminary
%sections. 
%}}}}   

\begin{definition}      
A connected expression $e$ is \concept{action-live} if for all actions
$a$ in $\Sigma$, from any reachable derivative of $e$, we can reach an
$a$-derivative of $e$. 
\end{definition}  

%\fbox{\parbox{\linewidth}{\textcolor{red}{\textsf{ 
%give an example of ce-cables having action-liveness and one which does not have this
%property.  
%}}}}   
 
\subsection{Relating connected-expressions with pairings and with cables} 
\label{lab:subsec:pdp2pair} 

First, we show how connected expressions with equal choice and
consistent pairings can be seen as connected expression with cables
and having  equal source and product-derivatives property. 
 
\begin{theorem}
\label{lab:tm:pair2pdp}
Let $\Sigma$ be a distributed alphabet and $e$ be a connected
expression having equal choice and consistent pairing of actions, defined over $\Sigma$.  
Then for the language of $e$, we can construct a connected expression
$e'$ with cables having equal source and product-derivatives property.  
The constructed expression $e'$ with cables is exponential
in the size of expression $e$ with pairings. 
\end{theorem}  
 
\begin{proof}%[of Theorem~\ref{lab:tm:pair2pdp}]        
Let $e=(s_1,\ldots,s_k)$ be a connected expression with pairings.
We take $e'= (s'_1=s_1,\ldots,s'_k=s_k)$ as our connected expression
with cables. We construct $\cables(a)$ relation for each action $a$ in
Let $loc(a)=\{1,\ldots,m\}$.  
In each $s_i$ and for all $a$ in $\Sigma$, for each block $B$ in $\Part_a(s_i)$  
we consider the set $X$-post-effects($B$) as the set of post-effects of $B$  
which are mutually disjoint.  
%Let $loc(a)=\{1,\ldots,m\}$.  
We build the set of $a$-cables $\cables(a)$ for $e'$ from the given
pairing relation $\pairing(a)$ of $e$ as follows. 
 
\[\begin{array}{l}
        \cables(a)= \{( (B_1,\ldots,B_m),(E_1,\ldots,E_m)) \mid 
                (B_1,\ldots,B_m) \in \pairing(a) ~\mbox{where}~\\
	\hspace*{0cm}(B_i,E_i)~\mbox{is an a-duct of}~
s_i,~\mbox{and} ~E_i ~\mbox{is in}~X\mbox{-post-effects}(B_i), 
	~\mbox{for all}~ i~\in loc(a) \}.%\{1,\ldots,m\}%~\rmand~i \in loc(a). 
%\rmand ~\mbox{importantly}~ (p_1,\ldots,p_m)~\\
%	\hspace*{0.51cm} \mbox{is part of some reachable marking of}~$A$. 
    \end{array}
\]
See that each block $B$ appears in at least one $a$-cable as it appears
in at least one tuple of $\pairing(a)$, and when $B$ appears more than once
i.e., it appears in two distinct $a$-ducts then the post-effects are
disjoint. Therefore, the relation constructed above satisfies
definition of cables relation. 
If we have $k_i$ mutually exclusive post-effects for $B_i$, for each
$i$ in $loc(a)$, then we have $k_1 \times k_2 \times k_m$ cables in
$\cables(a)$ having $B_1,\ldots,B_m$ as their set of pre-blocks,  
which is exponential in the number of locations. 
%and also in the number of pairing relation. 

%\noindent 
(\emph{Proof of $e'$ having equal source property}): 
Assume not i.e., we  have an $a$-cable $c$ and a $b$-cable $c'$ having
at least one common pre-block but their sets of pre-blocks are not equal. 
Without loss of generality we assume that $loc(a)=\{1,2\}$. 
Let $c= ( (B_1,E_1), (B_2,E_2))$ and $c'= ( (B_1,E'_1), (B_3,E_3))$, where 
$B_2 \neq B_3$. So we must have had $(B_1,B_2)$ in $\pairing(a)$ and $(B_1,B_3)$ in
$\pairing(b)$ of expression e with pairings. 
Therefore, init actions of $B_1$ includes $a$  and $b$.
We know that $e$ has equal choice pairing,  therefore
the sets of init actions of $B_1$ and $B_2$ are equal, hence 
$b$ is also an init action of $B_2$, which means that $B_2$ is also a block
in $\Part_b(s_2)$. Let $E'_2$ be one of the post-efffects of $B_2$
with respect to action $b$. 

Now consider a reachable derivative $e''=(r_1,r_2,\ldots,r_n)$ of $e$ in
which blocks $B_1$ and $B_2$ appear i.e., $r_1$ is in $B_1$ and $r_2$
is in $B_2$. 
We take an $a$-derivative $r'_1$ of $r_1$  and $r'_2$ of $r_2$ respectively, so
that $r'_1$ is in $E_1$ and $r'_2$ is in $E_2$.  
We also have a b-derivative of  $r''_1$ of $r_1$ in $E'_1$. 
But because of equal choice of
$B_1$ and $B_2$ we must also have a $b$-derivative $r''_2$ of
$r_2$ which reside in $E'_2$. 
Since connected expression  e has consistent pairing of actions,
$(B_1,B_2)$ must appear in $\pairing(b)$. 
This is a contradiction as we have $B_1$ appears twice in
$\pairing(b)$ relation, and is paired with distinct blocks of
$\Part_b(s_2)$.

%\noindent 
(\emph{Proof of $e'$ having product-derivatives property}):
%Let $loc(a)=\{1,\ldots,m\}$.  
Let ES-$\cables(a)$ be a set of $a$-cables having equal source property 
and let $(e_1,\ldots,e_m)$ be a tuple in \postdecomp(ES-$\cables(a)$). 
To show that $e'$ has product-derivatives property we have to show that
$(e_1,\ldots,e_m)$ is in \post(ES-$\cables(a)$) i.e., we have to prove 
existence of some $a$-cable $c$ in ES-$\cables(a)$ such that $\pi(c)=
(e_1,\ldots,e_m)$. Since $(e_1,\ldots,e_m)$ is in
\postdecomp(ES-$\cables(a)$),  
there exist $a$-cables $c_1,\ldots,c_m$ such that $\postset{c_i[i]}=e_i$
for all $i$ in $loc(a)$.  Hence we must have $a$-ducts $c_i[i]=(B_i,E_i)$ and 
$(B_1,\ldots,B_m) = \prodover_{i \in loc(a)} \preset{c_i[i]}$, the tuple
of source blocks of all cables in ES-$cables(a)$. 
By construction, it must be the case that $(B_1,\ldots,B_m)$ is in
$\pairing(a)$ of expression $e$ with pairings and $E_i$ are X-post-effects  
of $b_i$. Therefore, we must have constructed an $a$-cable  
$\prodover_{i \in loc(a)} (B_i,E_i)$, whose set of
post-effects is $\prodover_{i \in loc(a)} (E_i)$ as required.

%\noindent 
(\emph{Proof of $\Lang(e)=\Lang(e')$}):
Let $e=(s_1,\ldots,s_k)$ and $w \in \Lang(e)$.  
We prove that $w \in Lang(e')$ by doing induction on length of $w$.  
For each prefix of $w$ prove that if each $w$-derivative of $e$ is also 
an $w$-derivative of $e'$. 
Let $f=(u_1,\ldots,u_k)$ be $w$-derivative of $e$ and 
$f'=(v_1,\ldots,v_k)$ is in $\Der_a(f)$ therefore, $f'$ is in
$\Der_{wa} (e)$. 
Inductively, we assume that $f$ is also a valid $w$-derivative of $e'$ and
now we have to prove that $f'$ is a valid $a$-derivative of $f$ to show that
$f'$ is in $\Der_{wa}(e')$. Therefore we have to show that there exits an
$a$-cable whose pre-blocks appear in $f$ and whose post-effects appear in
$f'$. 

%Let $loc(a)=\{1,\ldots,m\}$. 
Since $e$ is consistent with pairing of actions, there exist blocks 
$B_i$ in $\Part_a(s_i)$ where each regular expression $u_i$ is in
$B_i$ for all $i$ in $loc(a)$, 
and $(B_1,\ldots,B_m)$ is a tuple in $\pairing(a)$. 
We have $v_i \in Der_a(u_i)$ for $i$ in $loc(a)$. 
Therefore, we must have X-post-effects $E_i$ of $B_i$ containing $u_i$ in
it. 
So, $(B_i,E_i)$ is an $a$-duct of $s_i$ for all $i$ in $loc(a)$. 
But $e'$ has product-derivatives property as proved above, 
therefore, for $B_1,\ldots,B_m$ appearing in $\pairing(a)$ and 
$E_1,\ldots,E_m$ their X-post-effects respectively, we have an $a$-cable 
$\prodover_{i \in loc(a)} (B_i,E_i)$. Since each $u_i$ is a component expression
of $f'$, each $E_i$ appears in $f'$. 
Therefore, $f'$ is a valid $a$-derivative of connected expression $f$
which completes the induction. 
\end{proof}  
 
Now we give a language preserving construction of connected expression with equal choice and
consistent pairings from a \cecables{} 
having equal source and product-derivatives property. 
 
\begin{theorem}      
\label{lab:tm:pdp2pair}
Let $\Sigma$ be a distributed alphabet and $e'$ be a connected expression  
with cables defined over it.%$\Sigma$.  
Then we can construct a connected expression $e$  
with pairings linear in the size of $e'$. 
And, 
%\begin{enumerate}  
%\item if $e'$ has equal source then $e$ has equal choice property, 
%\item in addition, if $e'$ is action-live then
%\begin{enumerate}  
%\item and if $e'$ has product moves property then $e$ has consistency of pairing.
%\item $Lang(e)=Lang(e')$.  
%\end{enumerate} 
%\end{enumerate}  
%\end{theorem}  
% 
if $e'$ has equal source then $e$ has equal choice property. 
In addition, if $e'$ is action-live then
if $e'$ has product moves property then $e$ has consistency of pairing.
$Lang(e)=Lang(e')$.  
\end{theorem}

\begin{proof}%[of Theorem~\ref{lab:tm:pdp2pair}]
Let $e'=(s'_1,\ldots,s'_k)$ be a connected expression with cables. 
We construct an expression with pairings $e=(s'_1=s_1,\ldots,s'_k=s_k)$ by
taking 
$\pairing(a)=\{\prodover_{i \in loc(a)} \preset{c[i]} \mid c \in
\cables(a)\}$, for each $a$ in $\Sigma$. 
The size of constructed $\pairing(a)$ relation is at most the
size of $\cables(a)$ relation for $e$.
Without loss of generality we assume that $loc(a)=\{1,2\}$. 

(\emph{Proof of $e$ having equal choice property}): 
Let $(B_1,B_2)$ be in $\pairing(a)$, and actions $a,b$ are
in $\Init(B_1)$, and action $a$ is in $\Init(B_2)$.  
Now we have to prove that $b$ is also present in $\Init(B_2)$. 
Since $(B_1,B_2)$ is in $\pairing(a)$ we must have had 
an $a$-cable $c$ with $B_1, B_2$ as its set of pre-blocks, so let 
$c=( (B_1,E_1),(B_2,E_2))$. 
And since $b$ is in $\Init(B_1)$, block $B_1$ is also present in the
$\Part_b(s_1)$, the set of partitions of $b$-derivatives of $s_1$. 
Let $E'_1$ be post-effects($B_1$) of $b$-derivatives of $B_1$. 
Then $(B_1,E'_1)$ is a $b$-duct of $s_1$. 
Therefore, we should have $b$-cable $c'$ having $c'[1]=(B_1,E'_1)$. 
Hence, we have two cables $c$ and $c'$ of expression $e'$ sharing a
pre-block $B_1$. But $e'$ has equal source property, therefore, $B_2$  
is also a pre-block of some $b$-cable $c'$ i.e., $B_2$ is a partition of  
$b$-derivatives of $s_2$ and hence it belongs to $\Part_b(s_2)$,  
implying that $b$ is present in $\Init(B_2)$ as required.

(\emph{Proof of $e$ having consistent pairing of actions}): 
%Without loss of generality assume that $loc(a)=\{1,2\}$. 
Let $e=(s_1,s_2,\ldots,s_k)$ with 
$f$ in $\Der_w(e)$ and $f'$ in $\Der_a(f)$ i.e., $f' \in \Der_{wa}(e)$. 
Inductively we assume that all intermediate derivatives excluding  $f$
are valid, i.e., their projections on $loc(a)$ appear in $\pairing(a)$
which we have constructed from $e$. 

Let $f=(r_1,r_2,\ldots,r_k)$ and $f'=(r'_1,r'_2,\ldots,r'_k)$,
so that $r'_k$ is in $\Der_a(r_i)$ for all $i$ in $loc(a)$. 
Let $B_j$ be the block in $\Part_a(s_i)$ containing $r_j$ 
and, $E_j$s be the post-effects of $B_j$ containing $r'_j$, 
for all $j$ in $\{1,\ldots,k\}$.

To complete the induction step, 
we have to prove that $f\project_{loc(a)}=(B_1,B_2)$ appears in
$\pairing(a)$. 
In the case that $c'=((B_1,E_1),(B_2,E_2))$ is already an $a$-cable in
$\cables(a)$ relation of $e'$ then we are done, because then by
construction of $\pairing(a)$ relation we have $(B_1,B_2)$ in
$\pairing(a)$. 

Now suppose that $c'$ is not an $a$-cable. 
%This also rules out the possibility of having any other $a$-cable
%having $B_1$ and $B_2$ as its pre-blocks. 
%\fbox{\parbox{\linewidth}{\textcolor{red}{\textsf{ 
%(why?).  
%}}}}   
Since $B_1$ is block of partition of $a$-derivatives of $s_1$, it 
must appear in at least one $a$-cable of $e'$ and similarly it holds for
$B_2$ also. 
Therefore, there must exist $a$-cables $c=((B_1,E_1),(B'_2,E'_2))$
and $d=((B'_1,E'_1),(B_2,E_2))$ in which $B_1$ and $B_2$ appear separately. 
At this point one could aruge that why $(B_1,E_1)$ is the chosen $a$-duct
to participate in $\cables(a)$ relation? We could have some $E''_1$ a
post-effect of $B_1$ containing $r'_1$ and that $(B_1,E''_1)$ could be
part of $a$-cable $c$ and not $(B_1,E_1)$. If $r'_1 \in E_1 \cap
E''_1$ then both $a$-ducts $(B_1,E_1)$ and  $(B_1,E''_1)$ can not
appear in $\cables(a)$ relation simultaneously. So at this point we could have chosen  
$(B_1,E''_1)$ and argument does not change, because all we require is $B_1$  
being a pre-block of some $a$-cable, which is guranteed by definition of $\cables(a)$  
relation. 

%\fbox{\parbox{\linewidth}{\textcolor{red}{\textsf{ 
%why $B_1$ has to appear with $E_1$ in $c$?
%it could appear with any other post-effect $E_3$ containg $r'_1$.
%}}}}  

%If $c'$ is an $a$-cable then using equal source property of $e'$ we  
%get $B_2=B'_2$ and similarly $B_1=B'_1$. 
If $c'$ is not $a$-cable then we have $B_2 \neq B'_2$, and hence 
$r_2 \notin B'_2$ and therefore, $f' \notin \Der_a(f)$.  
It also implies that  $B_1 \neq B'_1$, and hence 
$r_1 \notin B'_1$ and therefore, $f' \notin \Der_a(f)$. 
Hence $\Der_a(f)= \emptyset$. Importantly it implies that for 
any $w'$ containing action $a$, $\Der_{ww'}=\emptyset$. 
Therefore, we can not reach an $a$-derivative of $e'$ from $f$, which is in
contradiction with the fact that $e'$ is action-live. 
 
(\emph{Proof of $\Lang(e)=\Lang(e')$}):  
The direction of proving $\Lang(e') \subseteq \Lang(e)$ is easy because
each reachabale derivative of $e'$ is also valid derivative of $e$, since
$\pairing(a)$ relation is constructed from pre-blocks of cables in
$\cables(a)$ relation, for all $a$ in $\Sigma$. 
 
To prove that $\Lang(e) \subseteq \Lang(e')$, consider a word $w$ in
language of $e$, i.e., we reach some expression 
%$e''=(r''_1,\ldots, r''_k)$ in $\Der_w(e)$ having empty word in each of $r''_i$. 
$e''$ in $\Der_w(e)$ having empty word in the language of its each component expression.  
Let $w=ua$. 
%Without loss of generality we assume that $loc(a)=\{1,2\}$.  
Let $f$ be a $u$-derivative of $e'$ i.e., it is in $\Der_{u}(e')$. 
Inductively we assume that $f' \in \Der_{ua}(e)$. 
For induction step, we have to prove that $f' \in \Der_{ua}(e')$,
i.e., $f' \in \Der_{a}(f)$.
Let $f=(r_1,r_2,\ldots,r_k)$, and 
$f'=(r'_1,r'_2,\ldots,r'_k)$, where $r'_1$ (resp. $r'_2$) is an
$a$-derivative of $r_1$ (resp. of  $r_2$). 
Since $e$ is consistent with $\pairing(a)$, as proved above, 
$f\project_{loc(a)}$ appears in $\pairing(a)$.  
Let $f\project_{loc(a)}=(B_1,B_2)$. 
So by construction $e'$ must have cables $c_1$ and $c_2$ such that 
$\preset{c_1}[1]=B_1$ and $\preset{c_2}[2]=B_2$, but since they have equal
source, $\preset{c1}=\preset{c2}=(B_1,B_2)$. 
Let $E_1$ (resp. $E_2$) be the post-effect of $B_1$ (resp. $B_2$) containing 
$r'_1$ (resp. $r'_2$). 
Since $e'$ have product-derivatives property, 
there exist an $a$-cable $c$ with $B_1,B_2$ as its pre-blocks and 
$\postset{c}=(E_1,E_2)$, as required. 
 
\qed
\end{proof}

%Proofs from this subsection can be found in Appendix~\ref{lab:app:c0}. 
 
Expressions given in the following subsection correspond to product
systems with subset acceptance condition.  
 
\subsection{Sum of Connected Expressions (\sce)}
 \label{lab:subsec:sce} 
 
We give syntax for \name{sum of connected expressions (\sce)} defined over 
a distribution $(\Sigma_1,\Sigma_2,\ldots,\Sigma_k)$ of alphabet $\Sigma$.   
\[
e ::= \zero |e_1 + \cdots + e_m , \mbox{~where~}e_i \mbox{~is a
connected expression (\ce) over~} \Sigma
\]

For an \sce{} $e$ %defined over distributed alphabet
its semantics is given as follows:  
For the \sce{} $\zero$, we have $Lang(\zero)=\emptyset$.
For the \sce{} $e=e_1+\cdots+e_m$, its language is given as
%\[
$Lang(e)=Lang(e_1) \cup Lang(e_2)\cup \cdots \cup Lang(s_m)$. 
%Lang(e)=Lang(e_1) \cup Lang(e_2)\cup \ldots \cup Lang(s_m).
%\] 
The definitions of derivatives extended to \sce s is as given below.
The expression $\zero$ has no derivatives on any action.
Derivative of an \sce{} with respect to a letter $a$ is defined as: 
$\Der_a(e)=\Der_a(e_1) \cup \Der_a(e_2)\cup \cdots \cup \Der_a(s_m)  
$. Inductively $\Der_{aw}(e) = \Der_w(\Der_a(e))$.
%The set of all derivatives of $e$ is  
%    \[ 
%    \Der(e) = \cupover_{w \in \Sigma_i^*
%    \rmand i \in \{1,\ldots,m\}} \Der_w(e).
%    \]
The set of all derivatives $\Der(e)$ is union of sets of derivatives
    of $e$ over all words $w$ in $\Sigma_i^*$ for all $i$ in
    $\{1,\ldots,m\}$. 
    
    % = \cupover_{w \in \Sigma_i^*} \Der_w(e)$.

%\subsubsection{Sum of Connected Expressions with pairing (\scepairing) and their properties}
 
\begin{definition}[\sce{} with pairings)]       
An \sce{} $e=e_1 + \cdots + e_m$ where each $e_i$ is a \cepairings{}, 
is called a sum of connected expressions with pairing (\scepairings). 
An \scepairings{} $e$ is said to have equal choice property if each
component \cepairings{} $e_i$ of the sum also has it. 
\end{definition}  

\begin{example}
\label{lab:ex:scepairing} 
The expression $e=e_1+e_2$ where $e_1=\fsync((ab+ac)^*, (ad+ae)^*)$
is the \cepairings{} of Example~\ref{ex-ce1pairing} 
and $e_2=\fsync((ab+ac)^*a, (ad+ae)^*a)$ 
is the \cepairings{} of Example~\ref{ex-ce2pairing}
% (from Appendix~\ref{lab:app:d}),  
is an \scepairings{} with equal choice property, as both $e_1$ and $e_2$ have it. 
\end{example}

\begin{definition}[\sce{} with cables)]      
An \sce{} $e=e_1 + \cdots + e_m$ where each $e_i$ is a \cecables, then $e$ is called 
a sum of connected expressions with cables (\scecables). 
An \scecables{} $e$ is said to have \concept{equal source property} if each
component \cecables{} $e_i$ of the sum also has it. 
An \scecables{} $e$ has \concept{product-derivatives property} if each
component \cecables{} $e_i$ of the sum has it. 

\end{definition} 
 
\begin{example}
\label{lab:ex:scecables} 
The expression $e'=e_3+e_4$ where $e_3=\fsync((ab+ac)^*, (ad+ae)^*)$
is the \cecables{} of Example~\ref{ex-ce1pairing} and $e_4=\fsync((ab+ac)^*a, (ad+ae)^*a)$ 
is the \cecables{} of Example~\ref{ex-ce2pairing}.
%(from Appendix~\ref{lab:app:d}), is an \scecables{}.
 
As an example of how derivatives of 
\scecables{} from derivatives of \scepairings{} differ even while having identical
    components, see that $fsync(r_2,s_3) \in Der_a(e)$ (of
    Example~\ref{lab:ex:scepairing}) but it does not belong to $\Der_a(e')$. 
\end{example}

%% file: e2p.tex
%\section{Connected Expressions and Product Systems}
%\label{lab:section-exp-to-prod}

To get a product system with globals having subset-acceptance,  
from a sum of connected expression, we use an earlier result from \cite{P16},  
where construction of \psglobals{} with product-acceptance
was given from a \cecables{}. 

For each set of derivatives (pre-blocks and
after-effects), of a component regular expression, we constructed an unique state  
of a local component, which gives us product moves property in the constructed product  
system, if we have product-derivatives property for given \cecables{}. 
 
%Proofs from this section are given in Appendix~\ref{lab:app:c}.
%\subsection{From connected expressions to product systems}
 
\begin{lemma}[\cecables{} to \psglobals{} with product-acceptance \cite{P16}]
\label{lab:trm:cecables-to-pspac}
Let $e$ be a \cecables{}, defined over a distribution $\Sigma$. Then for  
the language of $e$, we can compute a \psglobals{} with product-acceptance linear in the size of expression $e$. 
Further, 
if $e$ had equal-source, then system $A$ has same source property; and,
if $e$ had product-derivatives then $A$ has product moves property.
\end{lemma}
 
Using Lemma~\ref{lab:trm:cecables-to-pspac}, we get \psglobals{} with subset-acceptance,  
from sum of connected expressions. 

\begin{theorem}[\sce-cables{} to \psglobals{} with subset-acceptance]
\label{lab:trm:scecables-to-pssac}
Let $e$ be an sum of connected expression defined over $\Sigma$.  
Then we can construct a \psglobals{} $A$ with subset-acceptance
for the language of $e$. 
If e had equal source property, then  has same source property.
In addition, if $e$ has product-derivatives property then 
$A$ has product moves property.
\end{theorem} 
\begin{proof}        
    Let $e=e_1+\ldots+e_m$ be an sum of \cecables{} having equal
    source and product-derivatives property. Language of expression e
    is $\Lang(e)=\Lang(e_1) \cup \ldots \cup \Lang(e_m)$. 
 
    Using Lemma~\ref{lab:trm:cecables-to-pspac}, we construct \psglobals{} $A_i$  
    with product-acceptance condition, for the language of each \cecables{} $e_i$  
    having the same source and product-derivatives property. That is 
    $\Lang(A_i)=\Lang(e_i)$ for all $i$ in $\{1,\ldots,m\}$. 
    
    Using language characterization of \psglobals{} with subset-acceptance conditions  
    given in Corollary~\ref{lab:cor:fcpsglobalssac-lang-char}, we get an \psglobals{} $A$  
    with subset-acceptance condition, over $\Sigma$ such that  
    $\Lang(A)=\Lang(A_1) \cup \ldots \cup \Lang(A_m)$. 
Since Corollary~\ref{lab:cor:fcpsglobalssac-lang-char}, preserves global moves of
component \psglobals{} and, as underlying \cecables{} had equal source
and product-derivatives property, we get same source property and 
product moves property for each of the component $A_i$. Hence $A$ has
both these properties as required. 
\qed 
\end{proof}

\begin{example}
For \scecables{} $e'$ of Example~\ref{lab:ex:scecables}, we can construct 
\psglobals{} with same source property and subset-acceptance $\calD$ of  
Example~\ref{lab:ex:psglobals}, accepting language $L_s$, using 
Theorem~\ref{lab:trm:scecables-to-pssac}.
\end{example}

%\subsection{From product systems to connected expressions }
A language preserving construction of connected expressions  
with equal source property, from a \psglobals{} with same source
property and product-acceptance, was given in~\cite{P16}. 
Since each local state--either a source state or a target state of
local move--is mapped uniquely to a set of derivatives of component
regular expression, we have product-derivatives property for the
expression, if the product system had product-moves property. 
 
\begin{lemma}[\psglobals{} with product-acceptance to \ce-cables{} \cite{P16}]       
\label{lab:tm:psglobalpac2cecables}
Let $\Sigma$ be a distributed alphabet and, $A$ be a product system with  
globals and product-acceptance, %having same source and product moves property,  
defined over $\Sigma$.  
For the language of $A$, we can construct a connected expression $e$
with cables, exponential in the size of the given product system. 
%Further, 
%\begin{itemize}  
%\item if product system has same source property then connected
%expression with cables has equal source property,
%\item in addition, if it has product-moves property then connected expression
%with cables 
%has product-derivatives property.
%\end{itemize}  
Furthermore, if product system has same source property then connected
expression with cables has equal source property,
in addition, if it has product-moves property then connected expression
with cables has product-derivatives property.
\end{lemma} 
 
Now using Lemma~\ref{lab:tm:psglobalpac2cecables}, we get a sum of connected
expressions for \psglobals{} with subset-acceptance.
 
\begin{theorem}[\psglobals{} with subset-acceptance to \sce-cables{}]       
\label{lab:trm:pssac-to-scecables}
Let $A$ be a product system with  
globals and subset-acceptance, %having same source and product moves property,  
defined over distribution $\Sigma$.  
For the language of $A$, we can construct a sum of connected expression $e$
with cables. And, if product system has same source property then $e$ has equal source property.
Also, in addition, if $A$ has product moves property then  
$e$ has product-derivatives property.
\end{theorem}  
\begin{proof}%[Theorem~\ref{lab:trm:pssac-to-scecables}]
    Let $A$ be an \psglobals{} with subset-acceptance condition and
same source property.
    Then using Corollary~\ref{lab:cor:fcpsglobalssac-lang-char}, there
exist \psglobals{} $A_1,\ldots,A_m$ with product-acceptance conditions
%and same source property,
such that 
    $\Lang(A)=\Lang(A_1) \cup \ldots \cup \Lang(A_m)$.
    Note that each $A_i$ has same source property. 
 
    For the language of each \psglobals{} $A_i$ with product-acceptance condition, 
    we can construct \cecables{} $e_i$ with equal source property, using  
    Lemma~\ref{lab:tm:psglobalpac2cecables}.

    From these we construct a sum of \cecables{} $e=e_1+\ldots+e_m$ which has equal
source property and language $\Lang(e_1) \cup \ldots \cup \Lang(e_m)$
which is $\Lang(A)$.
\qed 
\end{proof}

\begin{example}
For a \psglobals{} with same source property and subset-acceptance $\calD$ of  
Example~\ref{lab:ex:psglobals}, accepting language $L_s$, we can
construct an \scecables{} $e'$ of Example~\ref{lab:ex:scecables} using
Theorem~\ref{lab:trm:pssac-to-scecables}. 
\end{example}
 
%Proofs of Theorem~\ref{lab:trm:scecables-to-pssac} and Theorem\ref{lab:trm:pssac-to-scecables}  
%are given in Appendix~\ref{lab:app:c}.
 
Using equivalence of \psmatchings{} with product-acceptance and \cecables{},
%results  
from~\cite{PL14,PL15},  
and Corollary~\ref{lab:cor:fcpsmatsac-lang-char}
we get language equivalent 
\scepairings{} for \psmatchings{} with subset-acceptance, and
vice-versa.

\begin{theorem}[\psmatchings{} with subset-acceptance to \sce-pairings{}]       
\label{lab:trm:psmatsac-to-scepairings}
Let $A$ be a product system with conflict-equivalent and consistent
matchings, having subset-acceptance. %having same source and product moves property,  
For the language of $A$, we can construct a sum of connected expression $e$
with equal choice and consistent pairings.  
\end{theorem}  
 
%We have the reverse translation given as below. 
The converse result follows. 
 
\begin{theorem}[\sce-pairings{} to \psmatchings{} with subset-acceptance]       
\label{lab:tm:scepairings2psmatsac}
Let $\Sigma$ be a distributed alphabet and a sum of connected expression
$e$ defined over it, with equal choice and consistent pairings. Then 
for its language we can construct a product system with conflict-equivalent and consistent
matchings, having subset-acceptance. 
\end{theorem}

%% file: concl.tex
%\section{Conclusion}
%\label{lab:sec:concl}
In this paper, we have given a language ($L_s$ of Example~\ref{exfcnondcpnetpac})  
which can be accepted by \emph{free choice Zielonka automata}.
This language is not accepted by any synchronous product or direct
product, using Lemma~\ref{lab:lm:za-not-sp}, proof of which is
presented here, and was not given in \cite{P18}.
We have also given a language ($L_4$ of Example~\ref{lab:ex:fcns-dcp-sac}) which can be accepted
by a \emph{free choice synchronous product} and not by any direct product. 
A language which can be accepted by \emph{free choice direct product} ($L_3$ of  
Example~\ref{lab:ex:fcns-dcp-pac}) was given in~\cite{PL14}.
With this we have a hierarchy of labelled free choice nets similar to automata over distributed alphabets. 
In addition we have defined Zielonka automata with product
acceptance condition and its free choice restriction.
We have given language ($L_p$ of Example~\ref{lab:ex:fcns-pac}) of this class.
We used this intermediate automata to obtain Kleene theorem for \emph{free
choice Zielonka automata}. 
Lemma~\ref{lab:lm:zapac-not-dp} shows that this class is strictly  
more expressive than direct products with matching.  
%We believe that this class is strictly
%less expressive than Zielonka automata with same source and subset
%acceptance, but to prove this we need some language theoretic
%characterization of this class, which we do not have. 
In addition, ZA with same source has product moves property  
(class is not shown in the figure) then it is equivalent to SP with matching.
 
\begin{figure}[!htbp]
          \centering
          %\fbox{\input{./dafc-diamond-fig}}
          \input{./dafc-diamond-fig}
          \caption{Free choice net languages and automata over
distributed alphabets} 
          \label{lab:fig:dafc-diamond}
\end{figure}
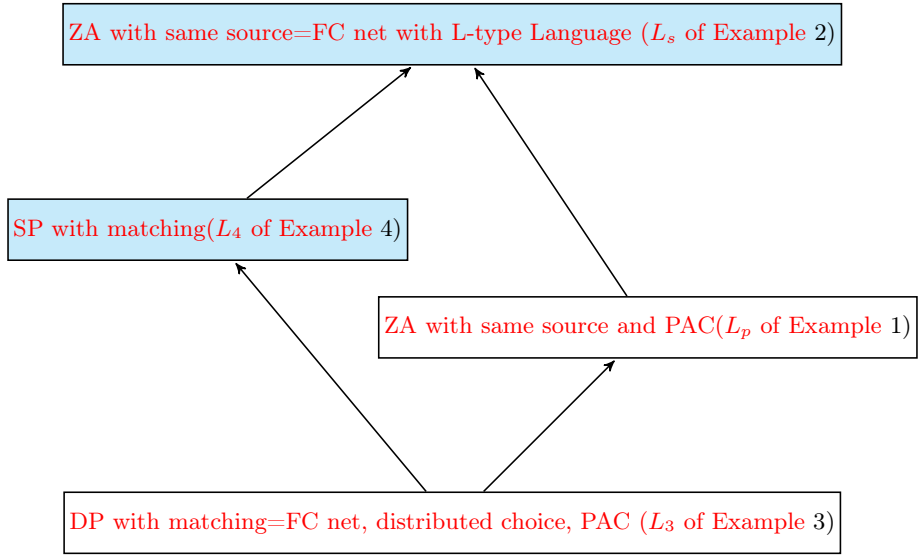

We give below the summary of correspondences established for the nets,
automata over distributed alphabets and expressions. 
To get an expression, for the language of a labelled $1$-bounded and  
S-coverable free choice net (with or without distributed choice)
having a finite set of final markings, we use Theorem~\ref{lab:tm:dcp2pmp} and
Theorem~\ref{lab:trm:pssac-to-scecables}.
In the reverse direction, we use Theorem~\ref{lab:trm:scecables-to-pssac} 
to get product sytem with subset acceptance from expressions and then  
Theorem~\ref{lab:tm:pmp2dcp} to get a language equivalent free choice net system. 
 
If the labelled free choice net has distributed choice and has a finite set of final markings, 
we have an alternate syntax for it. 
We first use Theorem~\ref{tm-dc-n2p2} to get equivalent product system, and then
Theorem~\ref{lab:trm:psmatsac-to-scepairings}, to get equivalent
expressions for the product system constructed. 
In the reverse direction, we use Theorem~\ref{lab:tm:scepairings2psmatsac} and then
Theorem~\ref{lab:tm:ps2n-dcp1}. 
All these correspondences and hirearchy of classes is shown in  Figure~\ref{lab:fig:dafc-diamond}. 
 
%In this paper, we have used product automata to get expressions for the
%Free choice nets, and given correspondences for all these three formalisms  
%for various classes. This kind of correspondence has been used~\cite{IorAnt12}  
%in concurrent code generation for discrete event systems.

%\textbf{Acknowedgements}: We thank anonymous referees of PNSE 2018 and ToPNoC, along with editors 
%Lucio Pomello and Lars Kristensen for their suggestions. 

%% file: dafc-diamond-fig.tex
%%%% diamaond figure for distributed alphabets. 
 
\begin{tikzpicture}[->,>=stealth',shorten >=1pt,%
auto,node distance=2.8cm,semithick,inner sep=2pt,bend angle=80,scale=0.65]
\tikzstyle{every state}=[fill=green!20,rectangle,draw=black!,text=red]
%\begin{tikzpicture}[->,>=stealth',shorten >=1pt,auto,node distance=2.8cm, inner sep=2pt]
%\draw[help lines] (0,0) grid(3,2);

\begin{scope}
  \node[state,fill=white!20]                     (dp) at(0,-2)   {DP
with matching=FC net, distributed choice, PAC ($L_3$ of  
Example~\ref{lab:ex:fcns-dcp-pac}) };
  \node[state,fill=cyan!20]                      (sp) at(-5,4)  {SP with matching($L_4$ of Example~\ref{lab:ex:fcns-dcp-sac}) };
  \node[state,fill=white!20]                     (zapac) at(4,2){ZA with same source and PAC($L_p$ of Example~\ref{lab:ex:fcns-pac}) };
  \node[state,fill=cyan!20]                      (za) at(0,8)   {ZA
with same source=FC net with L-type Language ($L_s$ of
Example~\ref{exfcnondcpnetpac})};

\path (dp) 	edge   (sp)
           	edge   (zapac)
      (sp) 	edge   (za) 
      (zapac)   edge  (za); 
\end{scope}
\end{tikzpicture}

%% file: ms.bbl
\begin{thebibliography}{10}
\providecommand{\url}[1]{\texttt{#1}}
\providecommand{\urlprefix}{URL }
\providecommand{\doi}[1]{https://doi.org/#1}

\bibitem{Ant96}
Antimirov, V.: Partial derivatives of regular expressions and finite automaton
  constructions. Theoret. Comp. Sci.  \textbf{155}(2),  291--319 (1996)

\bibitem{Brz64}
Brzozowski, J.A.: Derivatives of regular expressions. J. ACM  \textbf{11}(4),
  481--494 (1964)

\bibitem{DesEsp95}
Desel, J., Esparza, J.: Free choice {P}etri nets. Cambridge University Press,
  New York, USA (1995)

\bibitem{GarRag92}
Garg, V.K., Ragunath, M.: Concurrent regular expressions and their relationship
  to petri nets. Theoret. Comp. Sci.  \textbf{96}(2),  285--304 (1992)

\bibitem{Gra81}
Grabowski, J.: On partial languages. Fundam. Inform.  \textbf{4}(2),  427--498
  (1981)

\bibitem{Hac72}
Hack, M.H.T.: Analysis of production schemata by {P}etri nets. Project Mac
  Report TR-94, MIT (1972)

\bibitem{IorAnt12}
Iordache, M.V., Antsaklis, P.J.: The {ACTS} software and its supervisory
  control framework. In: Proceedings Conference on Decision and Control, {CDC}.
  pp. 7238--7243. IEEE (2012)

\bibitem{Jan87}
Jantzen, M.: Language theory of petri nets. In: ACPN. LNCS, vol.~254 (1987)

\bibitem{Lod06a}
Lodaya, K.: Product automata and process algebra. In: SEFM. IEEE (2006)

\bibitem{LMP11}
Lodaya, K., Mukund, M., Phawade, R.: Kleene theorems for product systems. In:
  DCFS, Proceedings. LNCS, vol.~6808. Springer (2011)

\bibitem{Mir66}
Mirkin, B.G.: An algorithm for constructing a base in a language of regular
  expressions. Engg. Cybern.  \textbf{5},  110--116 (1966)

\bibitem{Muk11}
Mukund, M.: Automata on distributed alphabets. In: D'Souza, D., Shankar, P.
  (eds.) Modern Applications of Automata Theory. World Scientific (2011)

\bibitem{Pet76}
Petersen, J.L.: Computation sequence sets. Journal of Computing and Systems
  Science  \textbf{13}(1),  1--24 (1976)

\bibitem{Pha-thesis}
Phawade, R.: Labelled Free Choice Nets, finite Product Automata, and
  Expressions. Ph.D. thesis, Homi Bhabha National Institute (2015)

\bibitem{P16}
Phawade, R.: Kleene theorems for labelled free choice nets without distributed
  choice. In: Cabac, L., Kristensen, L.M., R\"olke, H. (eds.) Proc. PNSE. CEUR
  Workshop Proceedings, vol.~1591, pp. 132--152. CEUR-WS.org (2016)

\bibitem{P18}
Phawade, R.: Kleene theorems free choice nets labelled with distributed
  alphabets. In: Daniel~Moldt, E.K., R\"olke, H. (eds.) Proc. PNSE. CEUR
  Workshop Proceedings, vol.~2138, pp. 77--98. CEUR-WS.org (2018)

\bibitem{PL14}
Phawade, R., Lodaya, K.: Kleene theorems for labelled free choice nets. In:
  Moldt, D., R\"olke, H. (eds.) Proc. PNSE. CEUR Workshop Proceedings,
  vol.~1160, pp. 75--89. CEUR-WS.org (2014)

\bibitem{PL15}
Phawade, R., Lodaya, K.: Kleene theorems for synchronous products with
  matching. Transactions on Petri nets and other models of concurrency
  \textbf{X},  84--108 (2015)

\bibitem{Zie87}
Zielonka, W.: Notes on finite asynchronous automata. Inform. Theor. Appl.
  \textbf{21}(2),  99--135 (1987)

\end{thebibliography}
